\documentclass[journal,draftclsnofoot,onecolumn,12pt,twoside]{IEEEtran}
\usepackage{graphicx}
\usepackage{amsmath,amssymb}
\usepackage{mathrsfs}
\usepackage{cases}
\usepackage{color}
\usepackage{amsfonts}
\usepackage{indentfirst}
\usepackage{slashbox}
\usepackage{cite}
\usepackage[justification=centering]{caption}
\usepackage{diagbox}
\usepackage{bm}
\usepackage{amssymb}
\usepackage{caption}
\usepackage{algorithm}
\usepackage{algpseudocode}
\usepackage{booktabs}
\usepackage{subfigure}
\usepackage{multirow}
\usepackage{amsthm}
\usepackage{chngpage}
\usepackage{epstopdf}
\usepackage{hyperref}
\usepackage{pstricks}
\usepackage{soul}
\usepackage{arydshln}
\makeatletter
\newtheoremstyle{mythm}{3pt}{3pt}{}{16pt}{\bfseries}{:}{.5em}{}
\theoremstyle{mythm}
\newtheorem{theorem}{Theorem}
\newtheorem{example}{Example}
\newtheorem{definition}{Definition}
\newtheorem{remark}{Remark}

\newtheorem{proposition}{Proposition}
\newtheorem{corollary}{Corollary}
\newtheorem{lemma}{Lemma}
\newtheorem{construction}{Construction}

\newtheorem{condition}{Condition}
\newcommand{\tabincell}[2]{\begin{tabular}{@{}#1@{}}#2\end{tabular}}
\begin{document}
\title{Novel Frameworks for Coded Caching  via Cartesian Product with Reduced Subpacketization
\author{Jinyu~Wang, Minquan~Cheng, Kai~Wan,~\IEEEmembership{Member,~IEEE,}
and~Giuseppe~Caire,~\IEEEmembership{Fellow,~IEEE}
\thanks{J. Wang and M. Cheng are with Guangxi Key Lab of Multi-source Information Mining $\&$ Security, Guangxi Normal University,
Guilin 541004, China  (e-mail: mathwjy@163.com, chengqinshi@hotmail.com). J. Wang is also with College of Mathematics and Statistics, Guangxi Normal University, Guilin 541004, China.}
\thanks{K. Wan and G. Caire are with the Electrical Engineering and Computer Science Department, Technische Universit\"{a}t Berlin,
10587 Berlin, Germany (e-mail: kai.wan@tu-berlin.de, caire@tu-berlin.de).  The work of K.~Wan and G.~Caire was partially funded by the European Research Council under the ERC Advanced Grant N. 789190, CARENET.}
}
}

\date{}
\maketitle

\begin{abstract}
Caching prefetches some library content at users' memories during the off-peak times (i.e., {\it placement phase}), such that the number of transmissions during the peak-traffic times (i.e., {\it delivery phase}) are reduced.
A coded caching strategy was originally proposed  by Maddah-Ali and Niesen (MN)   leading to a multicasting gain compared to the conventional uncoded caching, where each message in the delivery phase is useful to multiple users simultaneously. However, the MN coded caching scheme suffers from the high subpacketization which makes it impractical.
In order to reduce the subpacketization while retain the multicast opportunities in the delivery phase, Yan et al. proposed a combinatorial structure called placement delivery array (PDA) to design   coded caching schemes. In this paper, we propose two novel frameworks for constructing PDA via Cartesian product, which constructs a PDA for $mK_1$ users by the $m$-fold Cartesian product of a PDA for $K_1$ users.
By applying  the proposed frameworks to some  existing  PDAs, three novel caching schemes  are obtained which can significantly reduce the subpacketization of the MN scheme while slightly increasing the needed number of transmissions.
For instance,  for the third scheme which works for any number of users and any memory regime,     while reducing the coded caching gain by one,  the needed subpacketization   is at most $O\left(\sqrt{\frac{K}{q}}2^{-\frac{K}{q}}\right)$ of that of the MN scheme, where $K$ is the number of users, $0<z/q<1$ is the memory ratio of each user, and $q,z$ are coprime.

\end{abstract}

\section{Introduction}
Caching is an efficient technique to reduce peak-time traffic by prefetching contents near end-users during off-peak times, thereby reducing the transmission delay or equivalently increasing the bandwidth in communication networks. Conventional uncoded caching techniques focus on predicting the user demands for making an appropriate prefetching strategy, thus realizing a ``local caching gain", which scales with the amount of local memory \cite{BGW}. Maddah-Ali and Niesen (MN) showed that in addition to a local caching gain, coded caching can also attain a ``global caching gain", which scales with the global amount of memory in the network \cite{MN}, since each transmission of the MN scheme can serve a number of users simultaneously.

The MN coded caching problem considers a   $(K,M,N)$ caching system, where a server has a library of $N$ equal-sized files and broadcasts to $K$ users through an error-free shared-link. Each user has a cache of size $M$ files. A coded caching scheme consists of two phases:  {\it placement  } and {\it delivery}.
In the placement phase, each file is   divided into $F$ packets and each user's cache is filled without knowledge of the user's future demands. If packets are directly stored in each user's cache without coding, we call it uncoded placement. The quantity $F$ is referred to as the {\em subpacketization}.
In the delivery phase, each user  requests one arbitrary file and the server broadcasts coded packets such that each user can recover its desired file.
The worst-case load over all possible   demands is defined as the {\em transmission load} (or simply {\ load}) $R$, that is, the number of files that must be communicated so that any demands can be satisfied.
The coded caching gain of a coded caching scheme is defined as $\frac{K(1-\frac{M}{N})}{R}$, where $K(1-\frac{M}{N})$ is the load of the conventional uncoded caching scheme.
The objective of the MN coded caching problem is to minimize the load, i.e., maximize the coded caching gain.

The MN coded caching scheme, with a combinatorial design in the placement phase and a linear coding in the delivery phase, achieves a coded caching gain equal to $t+1$, where $t = \frac{K M}{N} $ is an integer.
If $\frac{K M}{N}$ is not an integer, then the lower convex envelope of the memory-load tradeoff can be achieved by memory sharing.
The achieved  load of the MN scheme is optimal under the constraint of uncoded placement and $N\geq K$~\cite{WTP2016,WTP2020}, and generally
 order optimal within a factor of $4$ \cite{GR}.  When a file is requested multiple times, by removing the redundant transmissions in the MN scheme, Yu, Maddah-Ali, and Avestimehr (YMA) designed a scheme which is optimal under the constraint of uncoded placement for $K>N$ \cite{YMA2018} and is order optimal to within a factor of $2$ \cite{YMA2019}. There are some other studies focusing on the load. For example, by using an optimization framework under a specific caching rule, Jin. et al. \cite{JCLC} derived the average load under conditions such as non-uniform file popularity and various demand patterns. When $K\leq N$ and all files have the same popularity, they also showed that the minimum load is exactly that of the MN scheme.

However, a main practical issue of the MN scheme is its   subpacketization  which   increases exponentially with the number of users.
A number of works in the literature were devoted to reduce the subpacketization of the MN scheme while increasing the load as the tradeoff.
The first scheme with lower supacketization compared with the MN scheme was proposed by Shanmugam et al. by dividing the users into groups and let the users in the same group have the same cache \cite{SJTLD}. Yan et al. \cite{YCTC} proposed a combinatorial structure called placement delivery array (PDA), which covers the MN scheme as a special case, referred to as the MN PDA.
Subsequently, Shangguan et al. \cite{SZG} showed that many previously existing coded caching schemes can be represented by PDA.
Based on
the concept of PDA, various caching schemes with lower subpacketzation than the MN scheme were proposed in~\cite{CJYT,CJWY,YCTC,CWZW}, some of which are listed in
  Table \ref{knownPDA}.
  Other combinatorial  construction structures   to reduce the subpacketization in the literature include  the linear block codes \cite{TR}, the special $(6,3)$-free hypergraphs \cite{SZG}, the $(r,t)$ Ruzsa-Szem\'{e}redi graphs \cite{STD}, the strong edge coloring of bipartite graphs \cite{YTCC},  the projective space \cite{K}, and other combinatorial designs such as \cite{ASK}.

\begin{table}
  \centering
  \caption{Some existing PDA schemes.  \label{knownPDA}}
  \begin{tabular}{|c|c|c|c|c|c|}
\hline
Scheme                          &\tabincell{c}{Number of\\ users $K$} & \tabincell{c}{Memory \\ratio $\frac{M}{N}$}& \tabincell{c}{Coded caching \\gain $g$}& Load $R$ & Subpacketization $F$ \\
\hline
\tabincell{c}{MN scheme in\cite{MN},\\ any $t=1,\ldots,k-1$}    & $k$  & $\frac{t}{k}$ & $t+1$ & $\frac{k-t}{t+1}$ & ${k\choose t}$\\
 \hline

 \tabincell{c}{Grouping method in \cite{SJTLD},\\any $n,k,t\in\mathbb{Z}^{+}$, $t<k$} &
 $nk$ &  $\frac{t}{k}$  & $t+1$ & $\frac{n(k-t)}{t+1}$ & ${k\choose t}$\\
\hline

 \multirow{2}{*}{\tabincell{c}{Partition PDA scheme\\in \cite{YCTC}, any $n,k\in\mathbb{Z}^{+}$}} &
 \multirow{2}{*}{$(n+1)k$} &
 $\frac{1}{k}$  & $n+1$ & $k-1$ & $k^n$\\ \cline{3-6}
 & & $\frac{k-1}{k}$& $(n+1)(k-1)$ & $\frac{1}{k-1}$ & $(k-1)k^n$ \\
\hline

 \multirow{2}{*}{\tabincell{c}{Scheme in \cite{SZG}, any \\$n,k,b\in\mathbb{Z}^{+}$, $b\leq n$}} &
 ${n\choose b}k^b$ &  $1-\left(\frac{k-1}{k}\right)^b$  & ${n\choose b}$ & $(k-1)^b$ & $k^n$\\ \cline{3-6}
 &&  $1-\frac{1}{k^b}$ & $(k-1)^b{n\choose b}$ & $\frac{1}{(k-1)^b}$& $(k-1)^bk^n$ \\

 \hline

\multirow{2}{*}{\tabincell{c}{Scheme in \cite{CJYT}, \\any $n,k,t,b\in\mathbb{Z}^{+}$,\\ $t<k$, $b\leq n$}} &
 $(n+1)k$ &  $\frac{t}{k}$  & $(n+1)\lfloor\frac{k-1}{k-t}\rfloor$ & $\frac{k-t}{\lfloor\frac{k-1}{k-t}\rfloor}$ & $\lfloor\frac{k-1}{k-t}\rfloor k^n$\\ [5pt] \cline{2-6}
 & ${n\choose b}k^b$ & $1-\left(\frac{k-t}{k}\right)^b$ & $\lfloor\frac{k-1}{k-t}\rfloor^b{n\choose b}$ & $\left(\frac{k-t}{\lfloor\frac{k-1}{k-t}\rfloor}\right)^b$& $\lfloor\frac{k-1}{k-t}\rfloor^bk^n$\\ [5pt]

 \hline

\multirow{2}{*}{\tabincell{c}{Scheme in \cite{TR},\\ any $n,l\in \mathbb{Z}^+$,\\ some $x,k\in \mathbb{Z}^+$ with\\ $n<l$ and $(n+1)|lx$}} & \multirow{2}{*}{$lk$} & $\frac{1}{k}$ & $n+1$ & $\frac{l}{n+1}(k-1)$ & $k^nx$\\ [11pt]\cline{3-6}
&& \tabincell{c}{$1-\frac{n+1}{lk}$\\ $\qquad$}  & $l(k-1)$ & $\frac{n+1}{l(k-1)}$& $(k-1)k^n\frac{lx}{n+1}$\\ [11pt]
\hline


\hline
   \end{tabular}

\end{table}

\subsection{Contributions}
In this paper, we propose two frameworks for constructing PDA via Cartesian product.
As applications of the frameworks, three new coded caching schemes are obtained and their performance analysis is subsequently provided.  Specifically, the following results are obtained.
\begin{itemize}
\item {\bf Novel frameworks for constructing PDA:} Given a PDA $\mathbf{P}$ with number of users $K_1$, coded caching gain $g$ and subpacketization $F_1$, if it satisfies some special conditions, i.e., Condition \ref{proper1} with parameter $\lambda$ given in Section \ref{sub:ConstructionPDATH1}, for any positive integer $m$, a new PDA with number of users $mK_1$ can be obtained by Construction \ref{constr1} in Section \ref{sub:ConstructionPDATH1} via Cartesian product, whose  memory ratio is the same as $\mathbf{P}$, coded caching gain is  $mg$ and subpacketization is $\lambda (\frac{F_1}{\lambda})^m$.
    Otherwise, if it satisfies some weaker conditions, i.e., Condition \ref{proper2} in Section \ref{sub:Applications of TH1}, which all previously existing PDAs satisfy,  for any positive integer $m$, a new PDA with number of users $mK_1$  can be obtained by Construction \ref{constructpmk} in Section \ref{sub:Applications of TH1} via Cartesian product, whose memory ratio is the same as $\mathbf{P}$, coded caching gain is $m(g-1)$ and subpacketization is $(g-1)F_1^m$.


 Consequently, if we want to construct a class of PDAs with good performance so then obtain a good coded caching scheme, we only need to find a good base PDA. Moreover, if the base PDA $\mathbf{P}$ has the minimum load, such as the MN PDA, a class of PDAs with approximate minimum load can be obtained.

\item {\bf Applications of the novel frameworks:} Three new coded caching schemes are obtained by using the above frameworks. The first scheme, referred to as Scheme A,  has smaller load and much smaller subpacketization than the schemes in \cite{SJTLD,CJYT} under some parameters. The second scheme,
    referred to as Scheme B, has much smaller subpacketization than the grouping method in \cite{SJTLD} when their loads are almost the same; it also has smaller load than the scheme in \cite{CJYT} when their subpacketizations are of the same level.
    Most of all, for any number of users $K$ and for any memory ratio $\frac{M}{N}=\frac{z}{q}$, where $q$ and $z$ are coprime, the coded caching gain of the third scheme (referred to as Scheme C) is only one less than that of the $K$-user MN scheme, while the subpacketization of Scheme C is at most $O\left(\sqrt{\frac{K}{q}}2^{-\frac{K}{q}}\right)$ of that of the MN scheme.

\end{itemize}

The rest of this paper is organized as follows. In Section \ref{system}, the   $(K,M,N)$ caching system and the definition of PDA   are introduced. In Section \ref{sec_constr1}, two frameworks for constructing PDA via Cartesian product are proposed. In Section \ref{sec:applications},   three novel caching  schemes  with low subpakcketization are obtained by using the proposed   frameworks.
 Finally,   Section \ref{conclusion} concludes the paper while some proofs are provided in the Appendices.

{\bf Notations:} In this paper, the following notations will be used unless otherwise stated.
\begin{itemize}
\item 
$[a:b]:=\left\{ a,a+1,\ldots,b\right\}$.  
\item $|\cdot|$ denotes the cardinality of a set.
\item $mod(a,q)$ denotes the least non-negative residue of $a$ modulo $q$.
\item
$
<a>_q:=\begin{cases}
       mod(a,q)\ \ \ \text{if}\ \ mod(a,q)\neq 0,\\
       \ \ \ q \ \ \ \ \ \ \ \ \ \  \text{if}\ \ mod(a,q)=0.\\
       \end{cases}
$
\item $[a:b]_q:=\{<a>_q,<a+1>_q,\ldots,<b>_q\}$.
\item Let $B=\{b_1,b_2,\ldots,b_{n}\}$ be a  set with $b_1<b_2<\ldots<b_n$, for any $i\in[1:n]$, $B[i]$ denotes the $i$-th smallest element of $B$, i.e.,$B[i]=b_i$.
\item Let $T=\{(t_{1,1},t_{1,2}),(t_{2,1},t_{2,2}),\ldots,(t_{n,1},t_{n,2})\}$ be a set of $2$-length vectors with $t_{1,2}<t_{2,2}<\ldots<t_{n,2}$, for any $i\in[1:n]$, $T[i]$ denotes the element of $T$ whose second coordinate is the $i$-th smallest in the second coordinate of all elements, i.e., $T[i]=(t_{i,1},t_{i,2})$.
\item Let ${\bf a}$ be a vector with length $n$, for any $i\in[1:n]$, ${\bf a}[i]$ denotes the $i$-th coordinate of ${\bf a}$.
For any subset $\mathcal{T}\subseteq [1:n]$, ${\bf a}[\mathcal{T}]$ denotes a vector with length $|\mathcal{T}|$ obtained by taking only the coordinates with subscript $i\in \mathcal{T}$.
\item For any positive integers $n$ and $t$ with $t< n$, let ${[1:n]\choose t}=\{\mathcal{T}\ |\   \mathcal{T}\subseteq [1:n], |\mathcal{T}|=t\}$, i.e., ${[1:n]\choose t}$ is the collection of all $t$-sized subsets of $[1:n]$.
\item For any two vectors ${\bf x}$ and ${\bf y}$ with the same length, $d({\bf x},{\bf y})$ is the number of coordinates in which ${\bf x}$ and ${\bf y}$ differ.
\item For any array  $\mathbf{P}$ with dimension $m\times n$, $\mathbf{P}(i,j)$ represents the element located in the $i$-th row and the $j$-th column of $\mathbf{P}$.
\item Let $\mathbf{P}$  be an array composed of a specific symbol $``*"$  and $S$ integers, for any integer $a$, $\mathbf{P}+a$ denotes a new $F\times K$ array
which is obtained by adding each element in $\mathbf{P}$  by $a$, where by contention $*+a=*$.
\end{itemize}

\section{System Model}
\label{system}

In this paper, we focus on a   $(K,M,N)$ caching system, illustrated in  Fig. \ref{fig-origin-system}. A server containing $N$ equal sized files is  connected to $K$ users through an error-free shared-link. Each user has a cache memory of size $M$ files where $0\leq M \leq N$.
Denote the $N$ files and $K$ users by $\mathcal{W}=\{W_1,W_2,\ldots,W_{N}\}$ and $\mathcal{K}=\{1,2,\ldots,K\}$ respectively. For such a system, an $F$-division coded caching scheme consists of two phases \cite{MN}:
\begin{itemize}
\item {\bf Placement phase:} Each file $W_n$ is divided equally into $F$ packets,
i.e., $W_{n}=\{W_{n,j}|j=1,2,\ldots,F\}$. Then some packets (or coded packets) are cached by each user from the server. If packets are cached directly without coding, we call it uncoded placement.
  Let $\mathcal{Z}_k$ denote the contents cached by user $k$. The placement phase is done without knowledge of the demands in the delivery phase.

\item {\bf Delivery phase:} Each user requests one arbitrary file from the server. Assume that the demand vector is $\mathbf{d}=(d_1,d_2,\cdots,d_{K})$, which represents that user $k$ requests $W_{d_k}$, where $d_k\in[1:N]$ and $k\in [1:K]$. Then the server broadcasts coded messages of total size $R_{{\bf d}} F$ packets to the users, such that each user can recover its desired file.
\end{itemize}
\begin{figure}
\centering
\includegraphics[width=3in]{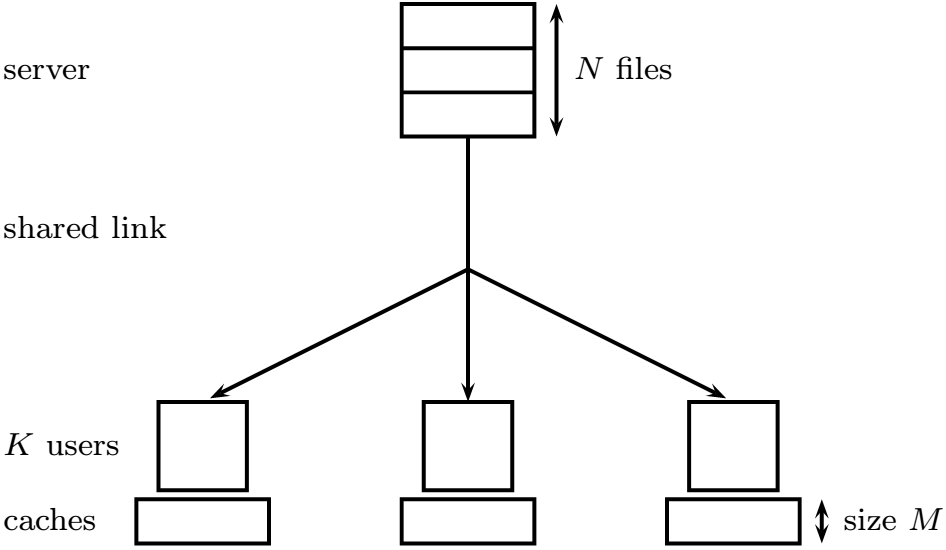}
\vskip 0.2cm
\caption{A $(K,M,N)$ caching system with $K=3,M=1,N=3$}\label{fig-origin-system}
\end{figure}
The  load of a coded caching scheme is defined as $R:=\max_{{\bf d} \in [1:N]^K} R_{{\bf d}}$,
  that is, the maximal normalized number of transmissions  over all possible  demands.

Placement delivery array proposed by Yan et al. \cite{YCTC} is a combinatorial structure, which designs coded caching schemes with uncoded placement and clique-covering delivery phase, i.e., each multicast  message transmitted in the delivery phase is a linear combination of some file packets and useful to a subset of users, each of which caches all  except one packets in this multicast message.
\begin{definition}(\cite{YCTC})
\label{def-PDA}
For  positive integers $K,F, Z$ and $S$, an $F\times K$ array  $\mathbf{P}$ composed of a specific symbol $``*"$  and $S$ integers in $[1:S]$, is called a $(K,F,Z,S)$ placement delivery array (PDA) if it satisfies the following conditions:
\begin{enumerate}
  \item [{\bf C1:}] The symbol $``*"$ appears $Z$ times in each column;
  \item [{\bf C2:}] Each integer in $[1:S]$ occurs at least once in the array;
  \item [{\bf C3:}] For any two distinct entries $\mathbf{P}(j_1,k_1)$ and $\mathbf{P}(j_2,k_2)$, if    $\mathbf{P}(j_1,k_1)=\mathbf{P}(j_2,k_2)=s\in[1:S]$, then $\mathbf{P}(j_1,k_2)=\mathbf{P}(j_2,k_1)=*$, i.e., the corresponding $2\times 2$  subarray formed by rows $j_1,j_2$ and columns $k_1,k_2$ must be of the following form
    \begin{eqnarray*}
    \left(\begin{array}{cc}
      s & *\\
      * & s
    \end{array}\right)~\textrm{or}~
    \left(\begin{array}{cc}
      * & s\\
      s & *
    \end{array}\right).
  \end{eqnarray*}

\end{enumerate}
\hfill $\square$
\end{definition}
A $(K,F,Z,S)$ PDA $\mathbf{P}$ is said to be a $g$-regular $(K,F,Z,S)$ PDA, $g$-$(K,F,Z,S)$ PDA or $g$-PDA for short, if each integer in $[1:S]$  appears exactly $g$ times in $\mathbf{P}$.

Based on a $(K,F,Z,S)$ PDA, an $F$-division coded caching scheme for the $(K,M,N)$ caching system, where $M/N=Z/F$,   can be obtained by using Algorithm \ref{alg:PDA}.
\begin{algorithm}[htb]
\caption{Coded caching scheme based on PDA in \cite{YCTC}}\label{alg:PDA}
\begin{algorithmic}[1]
\Procedure {Placement}{$\mathbf{P}$, $\mathcal{W}$}
\State Split each file $W_n\in\mathcal{W}$ into $F$ packets, i.e., $W_{n}=\{W_{n,j}\ |\ j\in[1:F]\}$.
\For{$k\in [1:K]$}
\State $\mathcal{Z}_k\leftarrow\{W_{n,j}\ |\ \mathbf{P}(j,k)=*, j\in[1:F], n\in [1:N]\}$
\EndFor
\EndProcedure
\Procedure{Delivery}{$\mathbf{P}, \mathcal{W},{\bf d}$}
\For{$s=1,2,\cdots,S$}
\State  Server sends $\bigoplus_{\mathbf{P}(j,k)=s,j\in[1:F],k\in[1:K]}W_{d_{k},j}$.
\EndFor
\EndProcedure
\end{algorithmic}
\end{algorithm}

\begin{example}
\label{ex1}
It is easy to verify that the following array is a $(4,4,2,4)$ PDA,
\begin{eqnarray}
\label{pda2}
\mathbf{P}_1=\left(\begin{array}{cccc}
*&*&3&1\\
2&*&*&4\\
1&3&*&*\\
*&2&4&*
\end{array}\right).
\end{eqnarray}

Based on the above PDA $\mathbf{P}_1$,  a $4$-division $(4,2,4)$ coded caching scheme can be obtained by using Algorithm \ref{alg:PDA}.
\begin{itemize}
   \item \textbf{Placement Phase}: From Line 2 of  Algorithm \ref{alg:PDA}, each file $W_n$ is split into $4$ packets, i.e., $W_n=\{W_{n,1},W_{n,2},W_{n,3},W_{n,4}\}$, $n\in [1:4]$. Then by Lines 3-5, the contents cached by each user are
\begin{eqnarray*}
       \mathcal{Z}_1=\left\{W_{n,1},W_{n,4}\ |\ n\in[1:4]\right\},
       \mathcal{Z}_2=\left\{W_{n,1},W_{n,2}\ |\ n\in[1:4]\right\},\\
       \mathcal{Z}_3=\left\{W_{n,2},W_{n,3}\ |\ n\in[1:4]\right\},
       \mathcal{Z}_4=\left\{W_{n,3},W_{n,4}\ |\ n\in[1:4]\right\}.
\end{eqnarray*}
   \item \textbf{Delivery Phase}: Assume that the request vector is $\mathbf{d}=(1,2,3,4)$. By Lines 8-10, the messages sent by the server are listed in Table \ref{table1}. Then each user can recover its requested file. For example, user $1$ requests the file $W_1=\{W_{1,1},W_{1,2},W_{1,3},W_{1,4}\}$ and has cached $W_{1,1}$ and $W_{1,4}$. At time slot $1$, it can receive $W_{1,3}\oplus W_{4,1}$, then it can recover $W_{1,3}$ since it has cached $W_{4,1}$. At time slot $2$, it can receive $W_{1,2}\oplus W_{2,4}$, then it can recover $W_{1,2}$ since it has cached  $W_{2,4}$. The load is $R=\frac{4}{4}=1$.
  \begin{table}[!htp]
  \centering
  \caption{Delivery steps in Example \ref{ex1} }\label{table1}
  \small{
  \begin{tabular}{|c|c|}
   \hline
   Time Slot& Transmitted Signnal\\ \hline
   $1$&$W_{1,3}\oplus W_{4,1}$\\ \hline
   $2$&$W_{1,2}\oplus W_{2,4}$\\ \hline
   $3$& $W_{2,3}\oplus W_{3,1}$\\ \hline
   $4$& $W_{3,4}\oplus W_{4,2}$\\ \hline
  \end{tabular}}
\end{table}
\end{itemize}
\hfill $\square$
\end{example}

A $(K,F,Z,S)$ PDA $\mathbf{P}$ is an $F\times K$ array composed of a specific symbol $``*"$ and $S$ integers, where columns represent the user indexes and rows represent the packet indexes. If $\mathbf{P}(j,k)=*$, then user $k$ caches the $j$-th packet of all files. The condition C1 of Definition \ref{def-PDA} implies that all the users have the same memory size and the memory ratio is $\frac{M}{N}=\frac{Z}{F}$. If $\mathbf{P}(j,k)=s$ is an integer, it means that the $j$-th packet of all files is not stored by user $k$. Then the server broadcasts a multicast message (i.e. the XOR of all the requested packets indicated by $s$) to the users at time slot $s$. The condition C3 of Definition \ref{def-PDA} guarantees that each user can recover its requested packet, since it has cached all the other packets in the multicast message except its requested one. The occurrence number of   integer $s$ in $\mathbf{P}$, denoted by $g_s$, is the coded caching gain at time slot $s$, since the coded packet broadcasted at time slot $s$ is useful for $g_s$ users. The condition C2 of Definition \ref{def-PDA} implies that the number of multicast messages transmitted by the server is exactly $S$, so the load is $R=\frac{S}{F}$.

\begin{lemma}(\cite{YCTC})
\label{th-Fundamental}If there exits a $(K,F,Z,S)$ PDA, there always exists an $F$-division $(K,M,N)$ coded caching scheme with the memory ratio $\frac{M}{N}=\frac{Z}{F}$ and   load $R=\frac{S}{F}$.
\hfill $\square$
\end{lemma}

\section{Frameworks for Constructing PDA via Cartesian Product}
\label{sec_constr1}
In this section, we propose two novel frameworks to construct a  PDA for $m K_1$ users based on the $m$-fold Cartesian product of some existing PDA for $K_1$ users. The grouping method in~\cite{SJTLD}   directly replicates any existing PDA for $K_1$ users by $m$ times, which results in a new PDA for
  $m K_1$ users  whose coded caching gain is exactly the same as that of the original PDA. In the contrast, by taking the $m$-fold Cartesian product, the resulting PDA in our construction has a coded caching gain is either equal to  $m g$  (if the original PDA satisfies Condition \ref{proper1}) or equal to $m(g-1)$ (if the original PDA satisfies Condition~\ref{proper2}), where $g$ is the coded caching gain of the original PDA.

  The rest of this section is organized as follows. We first list the main theorems of these two novel construction frameworks, and then provide the descriptions of the two construction frameworks in Section~\ref{sub:ConstructionPDATH1} and~\ref{sub:second framework}, respectively.



The first PDA construction  framework is given in the following theorem.
\begin{theorem}
\label{theopm}
For any positive integers $K_1,F_1,Z_1,S_1$ with $Z_1<F_1$, if there exists a $(K_1,F_1,Z_1,$ $S_1)$ PDA satisfying Condition \ref{proper1} with parameter $\lambda$ which will be given in Section~\ref{sub:ConstructionPDATH1}, then for any positive integer $m$, there always exists an $(mK_1,\lambda (\frac{F_1}{\lambda})^m, Z_1(\frac{F_1}{\lambda})^{m-1}, S_1(\frac{F_1}{\lambda})^{m-1})$ PDA, which leads to a $(K, M, N)$ coded caching scheme with number of users $K=mK_1$, memory ratio $\frac{M}{N}=\frac{Z_1}{F_1}$, subpacketization $F=\lambda (\frac{F_1}{\lambda})^m$ and load $R=\frac{S_1}{F_1}$.
\hfill $\square$
\end{theorem}

Condition \ref{proper1} has some strong constraints, which parts of existing PDAs do not satisfy. Hence,  we propose the following construction framework, where the original PDA should satisfy a weaker version of the constraints, i.e., Condition \ref{proper2} in Section~\ref{sub:second framework}, which all previously existing PDAs satisfy.
\begin{theorem}
\label{theopkm}
For any positive integers $K_1$, $F_1$, $Z_1$ and $S_1$ with $Z_1<F_1$, $g=\frac{K_1(F_1-Z_1)}{S_1}\in\mathbb{Z}$ and $g\geq 2$, if there exists  a $g$-$(K_1,F_1,Z_1,S_1)$ PDA satisfying Condition \ref{proper2} which will be given in Section~\ref{sub:second framework}, then for any positive integer $m\geq 2$, there always exists an $m(g-1)$-$(mK_1,(g-1)F_1^m, (g-1)Z_1F_1^{m-1},gS_1F_1^{m-1})$ PDA, which leads to a $(K,M,N)$ coded caching scheme with number of users $K=mK_1$, memory ratio $\frac{M}{N}=\frac{Z_1}{F_1}$, subpacketization $F=(g-1)F_1^m$ and load $R=\frac{g}{g-1}\frac{S_1}{F_1}$.
\hfill $\square$
\end{theorem}

\subsection{Construction of PDA in Theorem \ref{theopm}}
\label{sub:ConstructionPDATH1}

The Cartesian product of two sets $A$ and $B$ is defined as $A\times B=\{(x,y)|x\in A,y\in B\}$. Similarly, we define the Cartesian product of two arrays as follows.
\begin{definition}
\label{defcar}
Let $\mathbf{P}_1= \left(\begin{array}{c} {\bf a}_1 \\ {\bf a}_2 \\ \vdots \\    {\bf a}_{F_1}    \end{array}\right)$ be an $F_1\times K_1$ array,
 and $\mathbf{P}_2= \left(\begin{array}{c} {\bf b}_1 \\ {\bf b}_2 \\ \vdots \\    {\bf b}_{F_2}    \end{array}\right)
 $ be an $F_2\times K_2$ array. The Cartesian product of $\mathbf{P}_1$ and $\mathbf{P}_2$   is an $F_1F_2\times (K_1+K_2)$ array,  defined as
\begin{equation*}
\mathbf{P}_1\times \mathbf{P}_2=\left(\begin{array}{cc}
{\bf a}_1 & {\bf b}_1\\
\vdots         &  \vdots\\
{\bf a}_{F_1} & {\bf b}_1\\
\vdots &\vdots \\
{\bf a}_{1} & {\bf b}_{F_2} \\
\vdots         &  \vdots\\
{\bf a}_{F_1} & {\bf b}_{F_2}\\
\end{array}\right).
\end{equation*}
In particular, for any positive integer $m$, the $m$-fold Cartesian product of $\mathbf{P}_1$ is an $F_1^m\times mK_1$ array, defined as
\begin{equation*}
\mathbf{P}_1^{m}=\underbrace{\mathbf{P}_1\times \mathbf{P}_1\times\ldots\times\mathbf{P}_1}_{m}.
\end{equation*}
\hfill $\square$
\end{definition}

Next we will use an example to illustrate the main idea of the construction.
\begin{example}
\label{exm1}
Let us consider the PDA $\mathbf{P_1}$ in \eqref{pda2} again.
Based on $\mathbf{P}_1$, we will construct a $16\times 8$ PDA $\mathbf{P}_2$, as illustrated in Fig \ref{figexm1}.
The construction contains two steps. In the first step,   by taking the Cartesian product of $\mathbf{P}_1$ and $\mathbf{P}_1$, a $16\times 8$ array $\mathbf{P}_1\times \mathbf{P}_1$ can be obtained.
It can be seen that $\mathbf{P}_1\times \mathbf{P}_1$ does not satisfy Condition C3 of the definition of PDA (see Definition~\ref{def-PDA}); for example,   integer $1$ appears twice in the first row of $\mathbf{P}_1\times \mathbf{P}_1$, which contradicts Condition C3. So in the second step, we adjust the integers in $\mathbf{P}_1\times \mathbf{P}_1$  as follows, such that the resulting array satisfies Condition C3:
\begin{itemize}
\item In the first four columns of $\mathbf{P}_1\times \mathbf{P}_1$, there are exactly four replicas of $\mathbf{P}_1$. We replace the second, third, and last replicas by $\mathbf{P}_1+4$, $\mathbf{P}_1+8$, and $\mathbf{P}_1+12$ respectively.\footnote{\label{foot:recall P+4}  Recall that for any integer $a$, $\mathbf{P}_1+a$ denotes an array $(\mathbf{P}_1(j,k)+a)$, where $*+a=*$.}
\item For the last four columns, we replace $(2 \ 2 \ 2 \ 2)^{\top}$, $(3 \ 3 \ 3 \ 3)^{\top}$, $(4 \ 4 \ 4\ 4)^{\top}$, and $(1 \ 1 \ 1 \ 1)^{\top}$ by $( 2 \ 3 \ 4 \ 1 )^{\top}$, $( 6 \ 7 \ 8 \ 5 )^{\top}$, $( 10 \ 11 \ 12 \ 9 )^{\top}$, and $(14 \ 15 \ 16 \ 13 )^{\top}$, respectively.
\end{itemize}

\begin{figure}
  \centering
  \includegraphics[width=6in]{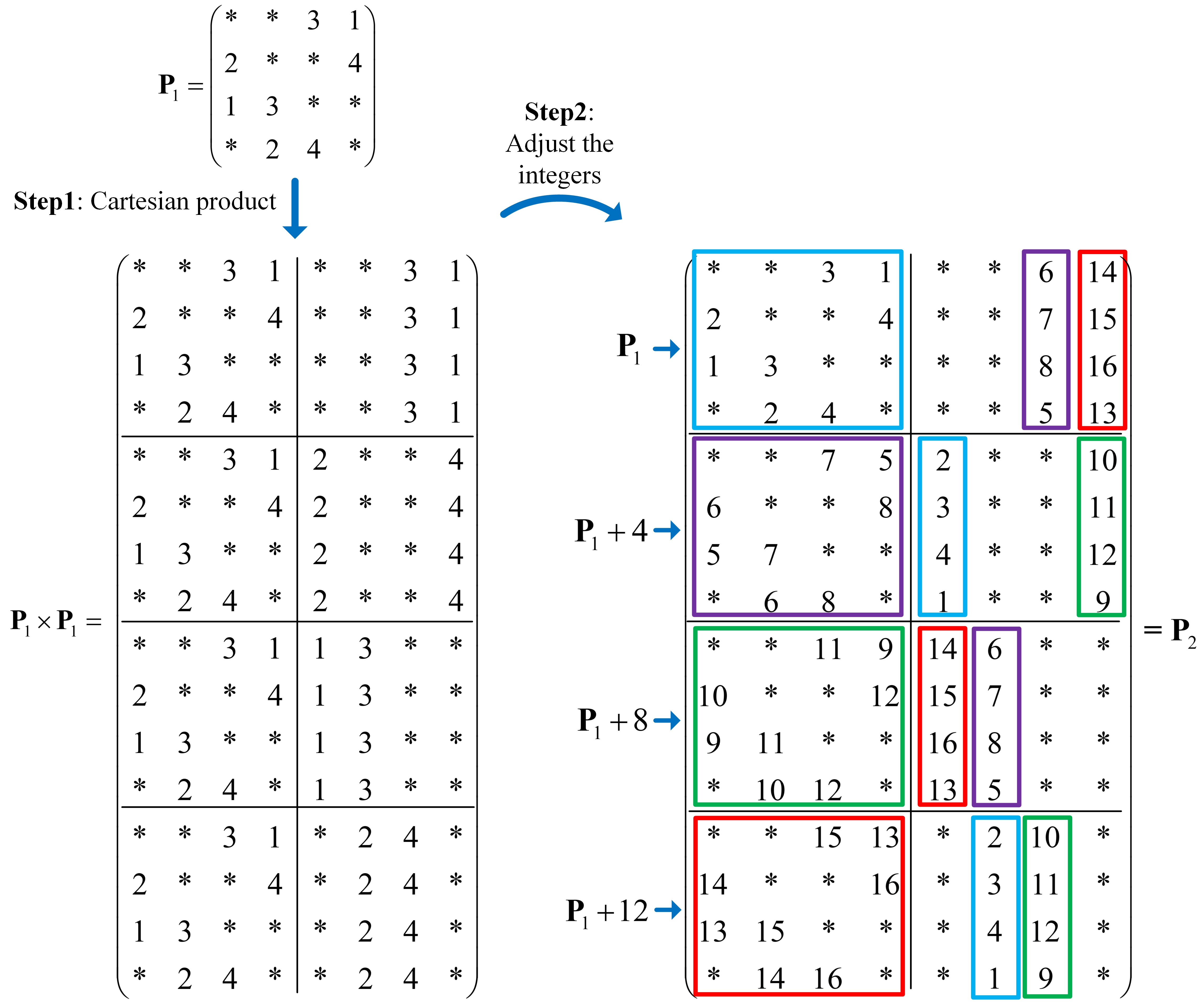}\\
  \caption{The process of generating $\mathbf{P}_2$ from $\mathbf{P}_1$.}\label{figexm1}
\end{figure}
The resulting array is denoted by $ \mathbf{P}_2$, which is an $(8,16,8,16)$ PDA and leads to an $8$-user coded caching scheme with the same memory ratio and load as the $4$-user coded caching scheme from $\mathbf{P}_1$.
\hfill $\square$
\end{example}

Inspired from Example \ref{exm1}, the intuition of our general construction for Theorem~\ref{theopm} could be explained as follows.
  Given a base PDA $\mathbf{P}_1$, for any positive integer $m$, we first take the $m$-fold Cartesian product of $\mathbf{P}_1$ (i.e., $\mathbf{P}_1^{m}$) and then adjust  the non-star entries in $\mathbf{P}_1^{m}$, such that the resulting array (denoted by $\mathbf{P}_m$) satisfies Condition C3 of PDA.
Thus we can obtain an $mK_1$-user coded caching scheme from the PDA $\mathbf{P}_m$, which has the   same memory ratio and load as
the $K_1$-user coded caching scheme from   $\mathbf{P}_1$.
	If the coded caching gain of the base PDA $\mathbf{P}_1$ is $g$, then the coded caching gain of $\mathbf{P}_m$ is  $mg$.

However, not all PDAs can serve as a base PDA in the above construction. In order to characterize the conditions  which a base PDA  should satisfy for our construction, we   introduce the following definitions.
\begin{definition}
For a $(K_1,F_1,Z_1,S_1)$ PDA $\mathbf{P}_1$ and an integer $s$, we say that the $i$-th row of $\mathbf{P}_1$ is a star row for $s$ if each column  of  $\mathbf{P}_1$ containing $s$  has  $*$ in  the  $i$-th row.
i.e., for any $\mathbf{P}_1(j,k)=s$ with $j\in[1:F_1]$ and $k\in[1:K_1]$, we have $\mathbf{P}_1(i,k)=*$. 
\hfill $\square$
\end{definition}

For example, let us focus on the PDA $\mathbf{P}_1$ in \eqref{pda2}. Integer $1$ appears twice in $\mathbf{P}_1$, i.e., $\mathbf{P}_1(3,1)=\mathbf{P}_1(1,4)=1$. Since $\mathbf{P}_1(4,1)=\mathbf{P}_1(4,4)=*$, the forth row of  $\mathbf{P}_1$ is a star row for   integer $1$. Similarly,
the $i$-th row of $\mathbf{P}_1$  is a star row for integer $i+1$, where $i\in \{1,2,3\}$.  So  each row of $\mathbf{P}_1$ is a star row for exactly one integer in $[1:4]$.

There may exist more than one star rows for an integer, in this case we need to assign exactly one star row  for this integer among its possible star rows; thus there may be some degree of freedom in defining a one-to-one correspondence between this integer and the corresponding star rows.
For a  $(K_1,F_1,Z_1,S_1)$ PDA $\mathbf{P}_1$, if we can assign exactly one star row for each integer in $[1:S_1]$ such that  each row of $\mathbf{P}_1$ is assigned to the same number of integers (i.e., to $S_1/F_1$ integers), then $\mathbf{P}_1$ can serve as a base PDA for our construction.

More generally, if $\mathbf{P}_1$ is vertically composed of $\lambda$ arrays (denoted by $\mathbf{P}_{1}^{(1)},\mathbf{P}_{1}^{(2)},\ldots,\mathbf{P}_{1}^{(\lambda)}$) each with dimension $\frac{F_1}{\lambda}\times K_1$ and the stars are placed in the same way in each $\mathbf{P}_{1}^{(i)}$ where $i\in[1:\lambda]$, and if we can assign a star row among the first $\frac{F_1}{\lambda}$ rows for each integer in $[1:S_1]$ such that each of the first $\frac{F_1}{\lambda}$ rows  is assigned to the same number of integers (i.e., to $\lambda S_1/F_1$ integers), then $\mathbf{P}_1$ can serve as a base PDA for our construction. It is worth noting that if $\lambda>1$, we will take the $m$-fold Cartesian product of each $\frac{F_1}{\lambda}\times K_1$ array $\mathbf{P}_{1}^{(i)}$ where $i\in[1:\lambda]$ instead of the $m$-fold Cartesian product of the entire $\mathbf{P}_1$.
Mathematically, any $(K_1,F_1,Z_1,S_1)$ PDA $\mathbf{P}_1$ satisfying the following constraint can serve as a  base PDA.
\begin{condition}
\label{proper1}
\begin{itemize}
\item [1)]There exists a positive integer $\lambda$ by which $F_1$ and $Z_1$ are dividable, such that $\mathbf{P}_1(j,k)=*$ if and only if $\mathbf{P}_1(<j>_{\frac{F_1}{\lambda}},k)=*$ for any $j\in[1:F_1]$ and $k\in[1:K_1]$.
\item[2)]There exists a mapping $\phi$ from $[1:S_1]$ to $[1:\frac{F_1}{\lambda}]$ such that:
\begin{itemize}
\item for each integer $s\in[1:S_1]$, the   $\phi(s)$-th row of $\mathbf{P}_1$ is a star row for  $s$;
\item for each   $j\in [1:\frac{F_1}{\lambda}]$, by defining $B_j =\{s|\phi(s)=j, s\in[1:S_1]\}$ in ascending order, it holds that  $|B_j|=\frac{\lambda S_1}{F_1}$.
\end{itemize}
\end{itemize}
\hfill $\square$
\end{condition}

Let us return to  PDA $\mathbf{P}_1$ in \eqref{pda2}, we will show that it satisfies Condition
\ref{proper1} with $\lambda=1$. $\mathbf{P}_1$ satisfies the first item of Condition \ref{proper1} obviously. Since there is only one star row for each integer, we have $\phi(1)=4,\phi(2)=1,\phi(3)=2,\phi(4)=3$; thus  $B_1=\{2\},B_2=\{3\},B_3=\{4\},B_4=\{1\}$. Consequently, $|B_j|=1$ for each $j\in[1:4]$. $\mathbf{P}_1$ satisfies the second item of Condition \ref{proper1}. So the PDA $\mathbf{P}_1$ in \eqref{pda2} satisfies Condition \ref{proper1}, then it can serve as a base PDA for our construction.
Let us go back to Example \ref{exm1} and explain Condition~\ref{proper1} intuitively. Focus on integer $1$ in  $\mathbf{P}_2$.  Since $\phi(1)=4$,  we have  $\mathbf{P}_2(4,1)=\mathbf{P}_2(4,4)=*$.
Similarly, by our construction, we have  $\mathbf{P}_2(8,1)=\mathbf{P}_2(8,4)=\mathbf{P}_2(16,1)=\mathbf{P}_2(16,4)=*$.  In addition, since integer $1$ is the smallest element of $B_4$ and the size of each $B_j$ is the same, the smallest integer for which the assigned star row is the first row must exist, i.e., $B_1[1]=2$, then we have $\mathbf{P}_2(1,1)=\mathbf{P}_2(1,2)=*$. Hence, by our construction, we have $\mathbf{P}_2(1,5)=\mathbf{P}_2(1,6)=\mathbf{P}_2(3,5)=\mathbf{P}_2(3,6)=*$. Hence, Condition~\ref{proper1} guarantees that $\mathbf{P}_2$ satisfies Condition C3  of Definition \ref{def-PDA}.

We are now ready to describe the construction for Theorem~\ref{theopm}. Recall that $B[i]$ denotes the $i$-th smallest element of the set $B$.
\begin{construction}
\label{constr1}
Given a $(K_1,F_1,Z_1,S_1)$ PDA $\mathbf{P}_1$ satisfying Condition \ref{proper1} with parameter $\lambda$, for any positive integer $m\geq 2$, let the row index set
\begin{equation}
\label{rowindex}
\mathcal{F}=\underset{i\in[1:\lambda]}{\bigcup}\left[(i-1)\frac{F_1}{\lambda}+1:i\frac{F_1}{\lambda}\right]^m
\end{equation}
 and the column index set $\mathcal{K}=[1:m]\times [1:K_1]$, a $\lambda (\frac{F_1}{\lambda})^{m}\times mK_1$ array
\begin{align*}
\mathbf{P}_{m}=\big(\mathbf{P}_{m}({\bf f},(\delta,b))| {\bf f}=(f_1,f_2,\ldots,f_{m})\in\mathcal{F}, (\delta,b)\in\mathcal{K} \big)
\end{align*}
 is defined as
\begin {align}
\label{constrPm}
\mathbf{P}_{m}({\bf f},(\delta,b))=\begin{cases}
 *    & \text{if}\ \mathbf{P}_1(f_{\delta},b)=*,\ \  \\
{\bf e} &\text{if} \ \mathbf{P}_1(f_{\delta},b)=B_l[\mu],\\
\end{cases}
\end{align}
where ${\bf e}=(e_1,e_2,\ldots,e_{m})$ such that
\begin {equation}
\label{constre}
e_h=\begin{cases}
\mathbf{P}_1(f_{\delta},b) &\text{if} \ \ h=\delta,\\
B_{<f_h>_{\frac{F_1}{\lambda}}}[\mu] & \text{otherwise}. \\
\end{cases}
\end{equation}
\hfill $\square$
\end{construction}


\begin{example}
For the PDA $\mathbf{P}_1$ in \eqref{pda2}, we have $\lambda=1$, $F_1=4$ and
$$B_1=\{2\},B_2=\{3\},B_3=\{4\},B_4=\{1\}.$$
When $m=2$, the array generated by Construction \ref{constr1} based on $\mathbf{P}_1$ in \eqref{pda2} is given in~\eqref{eqsubarray}.
\begin{figure}
\begin{align}
\label{eqsubarray}
\bordermatrix{%
         &(1,1) & (1,2) & (1,3) & (1,4) & (2,1) & (2,2) & (2,3) & (2,4)  \cr
(1,1)    &   *  &   *   & (3,2) & (1,2) &  *    &  *    & (2,3) & (2,1)  \cr
(2,1)    &(2,2) &   *   &   *   & (4,2) &  *    &  *    & (3,3) & (3,1)  \cr
(3,1)    &(1,2) & (3,2) &   *   &   *   &  *    &  *    & (4,3) & (4,1)  \cr
(4,1)    &  *   & (2,2) & (4,2) &   *   &  *    &  *    & (1,3) & (1,1)  \cr
(1,2)    &   *  &   *   & (3,3) & (1,3) & (2,2) &  *    &   *   & (2,4)  \cr
(2,2)    &(2,3) &   *   &   *   & (4,3) & (3,2) &  *    &   *   & (3,4)  \cr
(3,2)    &(1,3) & (3,3) &   *   &   *   & (4,2) &  *    &   *   & (4,4)  \cr
(4,2)    &  *   & (2,3) & (4,3) &   *   & (1,2) &  *    &   *   & (1,4)  \cr
(1,3)    &   *  &   *   & (3,4) & (1,4) & (2,1) & (2,3) &   *   &   *    \cr
(2,3)    &(2,4) &   *   &   *   & (4,4) & (3,1) & (3,3) &   *   &   *    \cr
(3,3)    &(1,4) & (3,4) &   *   &   *   & (4,1) & (4,3) &   *   &   *    \cr
(4,3)    &  *   & (2,4) & (4,4) &   *   & (1,1) & (1,3) &   *   &   *    \cr
(1,4)    &  *   &   *   & (3,1) & (1,1) &   *   & (2,2) & (2,4) &   *    \cr
(2,4)    &(2,1) &   *   &   *   & (4,1) &   *   & (3,2) & (3,4) &   *    \cr
(3,4)    &(1,1) & (3,1) &   *   &   *   &   *   & (4,2) & (4,4) &   *    \cr
(4,4)    &  *   & (2,1) & (4,1) &   *   &   *   & (1,2) & (1,4) &   * }
\end{align}
\end{figure}
When ${\bf f}=(1,1)$ and $(\delta,b)=(1,1)$, we have $\mathbf{P}_2({\bf f},(\delta,b))=*$ from \eqref{constrPm}, since $\mathbf{P}_1(f_{\delta},b)=\mathbf{P}_1(f_1,b)=\mathbf{P}_1(1,1)=*$. When ${\bf f}=(2,1)$ and $(\delta,b)=(1,1)$, since $\mathbf{P}_1(f_{\delta},b)=\mathbf{P}_1(f_1,b)=\mathbf{P}_1(2,1)=2=B_1[1]$, we have  $\mathbf{P}_2({\bf f},(\delta,b))=(\mathbf{P}_1(f_1,b),B_{<f_2>_{F_1}}[1])=(2,B_{1}[1])=(2,2)$ from \eqref{constrPm} and \eqref{constre}. The other entries in \eqref{eqsubarray} can be obtained from \eqref{constrPm} and \eqref{constre} similarly. By replacing the non-star entries in \eqref{eqsubarray} by integers according to the mapping $\psi$ illustrated in Table \ref{mpsi}, i.e., replacing ${\bf e}$ by $\psi({\bf e})$, the resulting array is exactly the array $\mathbf{P}_2$ in Fig. \ref{figexm1}.

\begin{table}
  \centering
  \caption{The mapping $\psi$}\label{mpsi}
  \begin{tabular}{|c|c|c|c|c|c|c|c|c|}
   \hline
   ${\bf e}$&$(1,2)$&$(2,2)$&$(3,2)$&$(4,2)$&$(1,3)$&$(2,3)$&$(3,3)$&$(4,3)$\\ \hline
   $\psi({\bf e})$&$1$&$2$&$3$&$4$&$5$&$6$&$7$&$8$ \\ \hline
   ${\bf e}$&$(1,4)$&$(2,4)$&$(3,4)$&$(4,4)$&$(1,1)$&$(2,1)$&$(3,1)$&$(4,1)$ \\ \hline
   $\psi({\bf e})$ &$9$&$10$&$11$&$12$&$13$&$14$&$15$&$16$ \\ \hline
  \end{tabular}
\end{table}
\hfill $\square$
\end{example}

By Construction \ref{constr1}, we have the following lemma, whose proof could be found in Appendix~\ref{prtheopm}.
\begin{lemma}
\label{lem:Theoerem 1 PDA}
The array $\mathbf{P}_{m}$ generated by Construction \ref{constr1} is an $(mK_1,\lambda \left(\frac{F_1}{\lambda})^m, Z_1(\frac{F_1}{\lambda})^{m-1},\right.$ $\left.S_1(\frac{F_1}{\lambda})^{m-1}\right)$ PDA.
\hfill $\square$
\end{lemma}
Hence, from Lemma~\ref{lem:Theoerem 1 PDA} we can prove Theorem \ref{theopm}.

If there exists a $g$-PDA satisfying Condition \ref{proper1} with parameter $\lambda$, the following result can be obtained from Theorem \ref{theopm}, whose proof could be found in Appendix~\ref{prcothm1}.
\begin{corollary}
\label{cothm1}
Given a $g$-$(K_1,F_1,Z_1,S_1)$ PDA  satisfying Condition \ref{proper1} with parameter $\lambda$ where $g=\frac{K_1(F_1-Z_1)}{S_1}$ is a positive integer,  there always exists an $mg$-$\left(mK_1,\lambda (\frac{F_1}{\lambda})^m, Z_1(\frac{F_1}{\lambda})^{m-1},\right.$ $\left.S_1(\frac{F_1}{\lambda})^{m-1}\right)$ PDA, for any positive integer $m$.
\hfill $\square$
\end{corollary}

In general, we need the exhaustive search to check whether a   PDA $\mathbf{P}_{1}$ satisfies Condition~\ref{proper1} or not. However, for some specific PDAs, such as the PDAs in \cite{YTCC} and \cite{CWZW}, we can prove that they satisfy Condition \ref{proper1} with $\lambda=1$ by their constructions.
\begin{proposition}
\label{prostrong}
For any positive integers $H$, $a$, $b$, $r$ satisfying $\max\{a,b\}<H$, $r<\min\{a,b\}$ and $a+b\leq H+r$, the ${H-a-b+2r\choose r}$-$\left({H\choose a},{H\choose b},{H\choose b}-{a\choose r}{H-a\choose b-r},{H\choose a+b-2r}{a+b-2r\choose a-r}\right)$ PDA in \cite{YTCC} satisfies Condition \ref{proper1} with $\lambda=1$ when $a=b+r$.
\hfill $\square$
\end{proposition}
\begin{proposition}
\label{proPDAOA}
For any positive integers $m,q,t$ with $t<m$ and $q\geq 2$, the ${m\choose t}$-$\left({m\choose t}q^t,q^{m-1},\right. \\ \left.q^{m-1}-(q-1)^tq^{m-t-1}, (q-1)^tq^{m-1}\right)$ PDA in \cite{CWZW} satisfies Condition \ref{proper1} with $\lambda=1$ when $t\geq 2$.
\hfill $\square$
\end{proposition}
The detailed proofs of Proposition \ref{prostrong} and \ref{proPDAOA}  could be found in Appendices
\ref{prprostrong} and \ref{prproPDAOA}, respectively.

\subsection{Construction of PDA in Theorem \ref{theopkm}}
\label{sub:second framework}
Given a PDA $\mathbf{P}$, if it satisfies Condition \ref{proper1}, for any positive integer $m$, Construction \ref{constr1} can be used to generate a PDA, which leads to an $mK_1$-user coded caching scheme with the same memory ratio and load as the $K_1$-user coded caching scheme from $\mathbf{P}$. However, Condition \ref{proper1} has some strong constraints, which some existing PDAs do not satisfy. In the following, based on a $g$-PDA $\mathbf{P}$ ($g\geq 2$) satisfying some weaker  constraints than Condition \ref{proper1}, which all previously existing PDAs satisfy, we propose a novel framework to construct a new PDA $\mathbf{P}_m$ by two steps. In the first step, we transform $\mathbf{P}$ into a $(g-1)$-PDA $\mathbf{P}_1$ satisfying Condition \ref{proper1} with $\lambda=g-1$. In the second step, based on the resulting PDA $\mathbf{P}_1$, we use Construction \ref{constr1} to generate an $m(g-1)$-PDA $\mathbf{P}_m$,  which leads to an $mK_1$-user coded caching scheme with the same memory ratio as the $K_1$-user coded caching scheme from $\mathbf{P}$.  The achieved load of the scheme from $\mathbf{P}_m$ is $\frac{g}{g-1}$ times of that of the scheme from $\mathbf{P}$.

Next we will use an example to illustrate the detailed steps.

\begin{example}
\label{ex2}
The following array  $\mathbf{P}$ is a $4$-$(5,10,6,5)$ PDA, which is exactly the $5$-user MN PDA with memory ratio $\frac{M}{N}=\frac{3}{5}$ and load $R=\frac{1}{2}$,
\begin{eqnarray}
\label{pda3}
\mathbf{P}=\left(\begin{array}{ccccc}
*&*&*&1&2\\
*&*&1&*&3\\
*&*&2&3&*\\
*&1&*&*&4\\
*&2&*&4&*\\
*&3&4&*&*\\
1&*&*&*&5\\
2&*&*&5&*\\
3&*&5&*&*\\
4&5&*&*&*\\
\end{array}\right).
\end{eqnarray}
Each row of $\mathbf{P}$ has three stars. For integer $1$, there is no star row, so $\mathbf{P}$ does not satisfy Condition~\ref{proper1}. Our construction contains the following two steps, as illustrated in Fig~\ref{transformP}.
\begin{figure}
  \centering
  \includegraphics[width=4.5in]{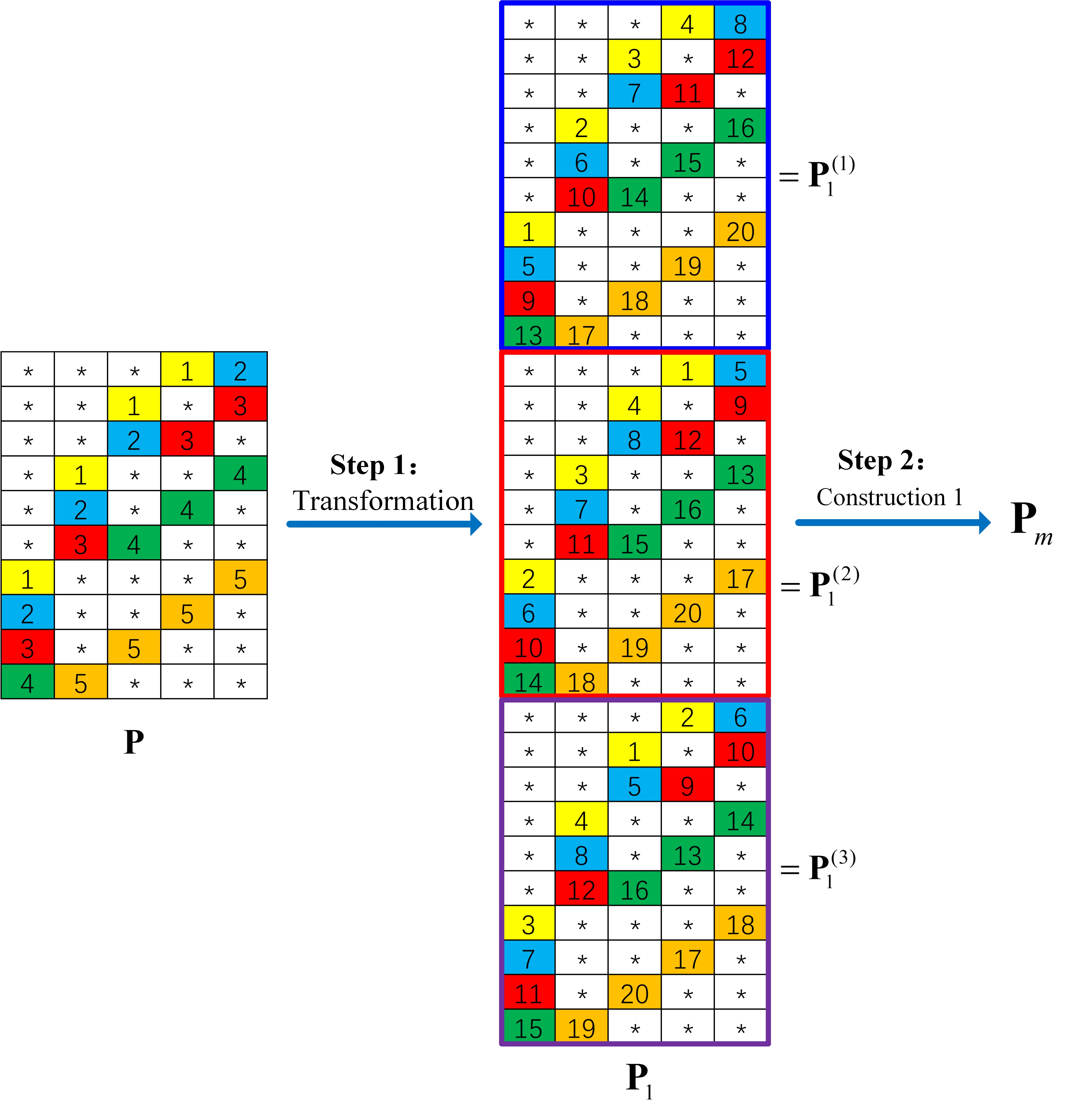}\\
  \caption{The process of generating $\mathbf{P}_m$ from $\mathbf{P}$ }\label{transformP}
\end{figure}

{\bf Step 1.} $\mathbf{P}$ is transformed into a $(g-1)$-PDA $\mathbf{P}_1$ with dimension $30\times5$ as follows. It can be seen that each integer $s$ appears  $4$ times in $\mathbf{P}$. We copy $\mathbf{P}$ three times. For integer $s \in [1:5]$,  we generate $4$ consecutive integers $(4(s-1)+1, 4(s-1)+2, 4(s-1)+3, 4s)$, and replace the  $4$ integer $s$'s in each replica by
a different rotation of  $(4(s-1)+1, 4(s-1)+2, 4(s-1)+3, 4s)$. The three updated replicas are denoted by $\mathbf{P}_1^{(1)}$, $\mathbf{P}_1^{(2)}$ and $\mathbf{P}_1^{(3)}$.
For example, we replace the  $4$ integer $1$'s in the first replica by $(1,2,3,4)$, from left to right, to generate $\mathbf{P}_1^{(1)}$; we replace the  $4$ integer $1$'s  in  the second replica by $(2,3,4,1)$ to generate $\mathbf{P}_1^{(2)}$; we replace the  $4$ integer $1$'s  in the third replica by $(3,4,1,2)$ to generate $\mathbf{P}_1^{(3)}$. By merging $\mathbf{P}_1^{(1)}$, $\mathbf{P}_1^{(2)}$ and $\mathbf{P}_1^{(3)}$ vertically into one array, the resulting array is denoted by $\mathbf{P}_1$.

 It is easy to verify that $\mathbf{P}_1$ is a $3$-$(5,30,18,20)$ PDA. Moreover, we will show that $\mathbf{P}_1$ satisfies Condition \ref{proper1} with $\lambda=3$. $\mathbf{P}_1$ satisfies the first item of Condition \ref{proper1} obviously.
 Let us focus on integer $1$ in $\mathbf{P}_1$. From Fig. \ref{transformP}, we have
  $$\mathbf{P}_1^{(1)}(7,1)=\mathbf{P}_1^{(2)}(1,4)=\mathbf{P}_1^{(3)}(2,3)=1.$$
 Since $1$ appears $4$ times in $\mathbf{P}$, there is exactly one position    (assumed to be $(p,q) \in [1:10] \times [1:5]$) where $\mathbf{P}(p,q)=1$ and none of  $\mathbf{P}^{(1)}_1(p,q),  \mathbf{P}^{(2)}_1(p,q),  \mathbf{P}^{(3)}_1(p,q)$ is equal to $1$. Obviously, $(p,q)=(4,2)$ in this example. We set the $p=4$-th row of  $\mathbf{P}_1$ as the star row for integer $1$.
 Similarly, we can assign exactly one star row among the first ten rows for each integer in $\mathbf{P}_1$, as shown in Table \ref{tablers}. Recall that for any $j\in[1:10]$, the set of integers for which the assigned star row is the $j$-th row, is denoted by $B_j$, which is illustrated in  Table \ref{tablebi}.
 From Table \ref{tablebi}, we have $|B_j|=2$ for each $j\in[1:10]$. $\mathbf{P}_1$ satisfies the second item of Condition \ref{proper1}. So $\mathbf{P}_1$ satisfies Condition \ref{proper1} with $\lambda=3$, then it can serve as a base PDA for Construction \ref{constr1}.
  \begin{table}
  \centering
  \caption{The mapping $\phi$ for assigning exactly one star row for each integer in $\mathbf{P}_1$  }\label{tablers}
  \small{
  \begin{tabular}{|c|c|c|c|c|c|c|c|c|c|c|c|c|c|c|c|c|c|c|c|c|}
   \hline
   $s$  &$1$&$2$&$3$&$4$&$5$&$6$&$7$&$8$&$9$&$10$&$11$&$12$&$13$&$14$&$15$&$16$&$17$&$18$&$19$&$20$\\ \hline
   The star row $\phi(s)$&$4$&$2$&$1$&$7$&$5$&$3$&$1$&$8$&$6$&$3$&$2 $&$9 $&$6 $&$5 $&$4 $&$10 $&$9 $&$ 8$&$7 $&$ 10$\\ \hline
  \end{tabular}}
\end{table}
 \begin{table}
  \centering
  \caption{The sets $B_1,B_2,\ldots,B_{10}$ for $\mathbf{P}_1$  }\label{tablebi}
  \small{
  \begin{tabular}{|c|c|c|c|c|c|c|c|c|c|}
   \hline
   $B_1  $  &$B_2$     &$B_3$    &$B_4$     &$B_5$     &$B_6$     &$B_7$     &$B_8$     &$B_9$     &$B_{10}$\\ \hline
   $\{3,7\}$&$\{2,11\}$&$\{6,10\}$&$\{1,15\}$&$\{5,14\}$&$\{9,13\}$&$\{4,19\}$&$\{8,18\}$&$\{12,17\}$&$\{16,20\}$\\ \hline
  \end{tabular}}
\end{table}

{\bf Step 2.} For any positive integer $m$, based on $\mathbf{P}_1$, Construction \ref{constr1} is used to generate a $3m$-$(5m,3\cdot10^m,18\cdot10^{m-1},2\cdot10^m)$ PDA $\mathbf{P}_m$, which leads to a $(K,M,N)$ coded caching scheme with number of users $K=5m$, memory ratio $\frac{M}{N}=\frac{3}{5}$, coded caching gain $g=3m$ and subpacketization $F=3\cdot10^m$.


It is worth noting that when $K=5m$ and $\frac{M}{N}=\frac{3}{5}$, the coded caching gain and subpacketization of the MN scheme are $g_{MN}=3m+1$ and $F_{MN}={5m\choose 3m}$, respectively. From Stirling's Formula $n!\approx \sqrt{2\pi n}\left(\frac{n}{e}\right)^n(n\rightarrow\infty)$, we have $F_{MN}=\frac{(5m)!}{(3m)!(2m)!}\approx \frac{\sqrt{5}}{\sqrt{12\pi m}}\left(\frac{5^5}{3^3\cdot2^2}\right)^m\approx\frac{\sqrt{5}}{\sqrt{12\pi m}}29^m$ when $m\rightarrow \infty$. Hence, for the same number of users and memory ratio, the coded caching gain of the proposed scheme is only one less than that of the MN scheme,  while we reduce the subpacketization  from $\frac{\sqrt{5}}{\sqrt{12\pi m}}29^m$ to $3\cdot10^m$.
\hfill $\square$
\end{example}

We are now ready to generalize our construction in Example~\ref{ex2}, which contains two steps.
\begin{itemize}
\item  {\bf Step 1}.  For any $g$-$(K_1,F_1,Z_1,S_1)$ PDA $\mathbf{P}$, we transform $\mathbf{P}$ into a $(g-1)$-$(K_1,(g-1)F_1,(g-1)Z_1,gS_1)$ PDA, denoted by $\mathbf{P}_1$. The construction of  $\mathbf{P}_1$ can be intuitively explained as follows. We copy $\mathbf{P}$ vertically $g-1$ times. For integer $s \in [1:S_1]$,  we replace the $g$ integer $s$'s in each replica by $g$ consecutive integers $(g(s-1)+1, g(s-1)+2,\ldots, gs)$ in a different rotation order, from left to right.
    Mathematically,   the $(g-1)F_1\times K_1$ array $\mathbf{P}_1$ is defined as
    \begin{equation}
    \label{eqPDAtoDPDA}
    \mathbf{P}_1(j,k)=\begin{cases}
    *, \ \ &\text{if} \ \ \mathbf{P}(<j>_{F_1},k)=*,\\
    g(s-1)+v, \ \ &\text{if} \ \ \mathbf{P}(<j>_{F_1},k)=s,T_s[\eta]=(<j>_{F_1},k),\\
    &i=\lceil\frac{j}{F_1}\rceil,v=\Phi_{i-1}((1,2,\ldots,g))[\eta],
    \end{cases}
    \end{equation}
    where
    \begin{equation}
    \label{Ts}
    \begin{array}{ccc}
    T_s&=&\{(j,k)|\mathbf{P}(j,k)=s,j\in[1:F_1],k\in[1:K_1]\} \\
    &=&\{(j_{s,1},k_{s,1}),(j_{s,2},k_{s,2}),\ldots, (j_{s,g},k_{s,g})\}\ \ \ \ \ \ \
    \end{array}
    \end{equation}
    with $k_{s,1}<k_{s,2}<\cdots<k_{s,g}$. $\Phi_i$ is a rotation function, defined as
   \begin{equation}
   \label{phi}
   \Phi_i((a_1,a_2,\ldots,a_{g}))=(a_{i+1},\ldots,a_{g},a_1,\ldots,a_{i})
   \end{equation}
   for any $g$-length vector $(a_1,a_2,\ldots,a_{g})$.  Recall that $T_s[\eta]$ denotes the element of $T_s$ whose second coordinate is the $\eta$-th smallest in the second coordinate of all elements, i.e., $T_s[\eta]=(j_{s,\eta},k_{s,\eta})$; ${\bf a}[\eta]$ denotes the $\eta$-th coordinate of the vector ${\bf a}$.

   Let us go back to Example \ref{ex2}, we have $F_1=10$ and $g=4$. Let us consider $\mathbf{P}_1(1,4), \mathbf{P}_1(11,4)$ and $\mathbf{P}_1(21,4)$ in Fig. \ref{transformP}. Since $<1>_{10}=<11>_{10}=<21>_{10}=1$, $\mathbf{P}(1,4)=s=1$ and $T_1=\{(j,k)|\mathbf{P}(j,k)=1,j\in[1:10],k\in[1:5]\}=\{(7,1),(4,2),(2,3),(1,4)\}$, we have $(1,4)=T_1[4]$. So from \eqref{eqPDAtoDPDA} we have $\mathbf{P}_1(1,4)=g(s-1)+v=4$, since $i=\lceil\frac{1}{10}\rceil=1$ and $v=\phi_{i-1}((1,2,3,4))[4]=\phi_{0}((1,2,3,4))[4]=(1,2,3,4)[4]=4$; $\mathbf{P}_1(11,4)=g(s-1)+v=1$, since $i=\lceil\frac{11}{10}\rceil=2$ and  $v=\phi_{i-1}((1,2,3,4))[4]=\phi_1((1,2,3,4))[4]=(2,3,4,1)[4]=1$; $\mathbf{P}_1(21,4)=g(s-1)+v=2$, since $i=\lceil\frac{21}{10}\rceil=3$ and $v=\phi_{i-1}((1,2,3,4))[4]=\phi_2((1,2,3,4))[4]=(3,4,1,2)[4]=2$.

We then introduce the following condition, which is a sufficient condition on   $\mathbf{P}$  for our construction.
\begin{condition}
\label{proper2}
Each row has the same number of stars.
\hfill $\square$
\end{condition}
With the above condition, we have the following proposition, which will be proved in   Appendix \ref{prproP1k}.
\begin{proposition}
\label{proP1k}
If $\mathbf{P}$ is a $g$-$(K_1,F_1,Z_1,S_1)$ PDA satisfying Condition \ref{proper2}, the array $\mathbf{P}_1$ defined in \eqref{eqPDAtoDPDA} is a $(g-1)$-$(K_1,(g-1)F_1,(g-1)Z_1,gS_1)$ PDA satisfying Condition \ref{proper1} with $\lambda=g-1$.
\hfill $\square$
\end{proposition}
\item {\bf Step 2.}    Since the PDA $\mathbf{P}_1$ defined in \eqref{eqPDAtoDPDA} can serve as a base PDA for Construction \ref{constr1}, for any positive integer $m$, based on $\mathbf{P}_1$,  Construction \ref{constr1} is used to generate a new PDA $\mathbf{P}_m$.
\end{itemize}

To summarize,   the construction for Theorem \ref{theopkm} is given as follows.
\begin{construction}
\label{constructpmk}
Given a $g$-$(K_1,F_1,Z_1,S_1)$ PDA $\mathbf{P}$ satisfying Condition \ref{proper2} where $g=\frac{K_1(F_1-Z_1)}{S_1}\in\mathbb{Z}$ and $g\geq 2$, for any positive integer $m\geq 2$, we construct $\mathbf{P}_m$ by Construction \ref{constr1}, based on $\mathbf{P}_1$ defined in \eqref {eqPDAtoDPDA}.
\hfill $\square$
\end{construction}

From Proposition \ref{proP1k} and Corollary \ref{cothm1}, we have the following lemma, from which Theorem \ref{theopkm} is proved.
\begin{lemma}
\label{lem:Theoerem 2 PDA}
The array $\mathbf{P}_{m}$ generated by Construction \ref{constructpmk} is an $m(g-1)$-$(mK_1,(g-1)F_1^m, (g-1)Z_1F_1^{m-1},gS_1F_1^{m-1})$ PDA.
\hfill $\square$
\end{lemma}

\section{Applications of the Novel PDA Construction Frameworks}
\label{sec:applications}
\subsection{Applications of Theorem \ref{theopm}}
\label{sub:Applications of TH1}
Two novel coded caching schemes are  obtained from Theorem \ref{theopm}.
First, when $a=b+r$, the PDA in \cite{YTCC} is an ${H-2b+r\choose r}$-$\left({H\choose b+r},{H\choose b},{H\choose b}-\right.$ $\left.{b+r\choose r}{H-b-r\choose b-r}, {H\choose 2b-r}{2b-r\choose b}\right)$ PDA satisfying Condition \ref{proper1} with $\lambda=1$ from  Proposition \ref{prostrong}. So the following result can be obtained from Corollary \ref{cothm1}.
\begin{corollary}
\label{strongbs}
For any positive integers $H$, $b$, $r$ and $m$ with $r<b\leq\frac{H}{2}$, there always exists an $m{H-2b+r\choose r}$-$\left(m{H\choose b+r},{H\choose b}^m,\right.$ $\left. \left({H\choose b}-{b+r\choose r}{H-b-r\choose b-r}\right){H\choose b}^{m-1}, {H-b\choose b-r}{H\choose b}^{m}\right)$ PDA, which leads to a $(K,M,N)$ coded caching scheme with number of users $K=m{H\choose b+r}$, memory ratio $\frac{M}{N}=1-\frac{{b+r\choose r}{H-b-r\choose b-r}}{{H\choose b}}$, subpacketization $F={H\choose b}^m$ and load $R={H-b\choose b-r}$.
\hfill $\square$
\end{corollary}

Second, for any positive integer $g$, we construct a $g$-$(\lceil\frac{g^2}{2}\rceil+g,\lceil\frac{g^2}{2}\rceil+g,\lceil\frac{g^2}{2}\rceil,\lceil\frac{g^2}{2}\rceil+g)$ PDA satisfying Condition \ref{proper1} with $\lambda=1$, by Construction \ref{contru2} in Appendix \ref{prmultibs}. So the following result can be obtained from Corollary \ref{cothm1}.

\begin{theorem}
\label{multibs}
For any positive integers $g$ and $m$, there always exists an $mg$-$(m(\lceil\frac{g^2}{2}\rceil+g),(\lceil\frac{g^2}{2}\rceil+g)^m, \lceil\frac{g^2}{2}\rceil(\lceil\frac{g^2}{2}\rceil+g)^{m-1},(\lceil\frac{g^2}{2}\rceil+g)^m)$ PDA, which leads to a $(K,M,N)$ coded caching scheme with number of users $K=m(\lceil\frac{g^2}{2}\rceil+g)$, memory ratio $\frac{M}{N}=\frac{\lceil\frac{g^2}{2}\rceil}{\lceil\frac{g^2}{2}\rceil+g}$, subpacketization $F=(\lceil\frac{g^2}{2}\rceil+g)^m$ and load $R=1$.
\hfill $\square$
\end{theorem}

Note that when $g=2$, the $2$-$(4,4,2,4)$ PDA generated by Construction \ref{contru2} is exactly the PDA in \eqref{pda2}; thus
the illustration of the   $8$-user coded caching scheme in  Theorem~\ref{multibs} could be found in Example~\ref{exm1}.



\subsection{Applications of Theorem \ref{theopkm}}
For any positive integers $q,z$ with $z<q$, the authors in \cite{YCTC} showed that the MN PDA is a $(z+1)$-$\left(q,{q\choose z},{q-1\choose z-1},{q\choose z+1}\right)$ PDA where the number of users $K=q$ and memory ratio $\frac{M}{N}=\frac{z}{q}$.  Since the MN PDA satisfies Condition \ref{proper2}, we can obtain the following coded caching scheme from Theorem \ref{theopkm}.
\begin{theorem}
\label{MNbs}
For any positive integers $q,z,m$ with $z<q$ and $m\geq 2$, there exists an $mz$-$\left(mq, z{q\choose z}^m, z{q-1\choose z-1}{q\choose z}^{m-1}\right.,$ $\left. (z+1){q\choose z+1}{q\choose z}^{m-1}\right)$ PDA, which leads to a $(K,M,N)$ coded caching scheme with the number of users $K=mq$, memory ratio $\frac{M}{N}=\frac{z}{q}$, subpacketization $F=z{q\choose z}^m$ and load $R=\frac{q-z}{z}$.
\hfill $\square$
\end{theorem}

When $q=5$ and $z=3$, the PDA in Theorem \ref{MNbs} is exactly the $3m$-$(5m,3\cdot10^m,18\cdot10^{m-1},2\cdot10^m)$ PDA $\mathbf{P}_m$ in Example \ref{ex2}.

It is worth noting that all previously existing PDAs are $g$-PDAs satisfying Condition \ref{proper2}. For any positive integer $m$, based on each of these PDAs, an $mK_1$-user coded caching scheme can be obtained from Theorem \ref{theopkm}, which has the same memory ratio  as the base PDA
and  the  load equal to $\frac{g}{g-1}$ times the load of the base PDA.
For the sake of simplicity, we do not  list them here.

\subsection{Performance Analysis}
\label{performance}
To avoid the heavy notations,
the novel schemes in Corollary \ref{strongbs}, Theorem \ref{multibs},   \ref{MNbs} are referred to as Schemes A, B, C respectively. In the following, we compare them with the existing PDA schemes listed in Table \ref{knownPDA}. Since the scheme in \cite{CJYT} is a generalization of the schemes in \cite{YCTC,SZG}, Schemes A and B are only compared with the schemes in \cite{SJTLD,CJYT,TR}.
Under the constraint of uncoded placement and $N\geq K$, the achieved  load of the MN scheme, denoted by $R_{MN}$, is optimal. For a scheme, if the ratio of the achieved load and the optimal load $R_{MN}$ tends to $1$ when the number of users tends to infinity, the scheme is called {\em asymptotically optimal}. The partition PDA scheme in \cite{YCTC} is asymptotically optimal.
So Scheme C is compared with the MN scheme in \cite{MN} and the partition PDA scheme in \cite{YCTC}.

\begin{itemize}
\item {\bf Scheme A in Corollary \ref{strongbs}.}
Since it is difficult to provide an analytic comparison, some numerical results are provided in Table \ref{tablecom}. It can be seen from Table \ref{tablecom}  that under some parameters, Scheme A has a lower load and a much lower subpacketization than the schemes in \cite{SJTLD,CJYT} for the same number of users and memory ratio. Compared to the scheme in \cite{TR}, Scheme A has a significant reduction on the subpacketization.
\begin{table}
  \centering
  \caption{Numerical comparison  of Scheme A in Corollary \ref{strongbs} with the schemes in \cite{CJYT,SJTLD,TR}  }\label{tablecom}
  \small{
  \begin{tabular}{|c|c|c|c|c|}
   \hline
     $K$ &  $\frac{M}{N}$ &Scheme & $R$ &  $F$ \\ \hline

   \multirow{4}*{$56m$} & \multirow{4}*{$\frac{13}{28}$} &\tabincell{c}{Scheme A in Corollary \ref{strongbs}\\  $H=8,b=3,r=2$} & $5$ &$56^m$ \\ \cline{3-5}

   &  &\tabincell{c}{Grouping method in \cite{SJTLD} and memory sharing \\$\frac{M}{N}=\frac{3}{7}:n=8,k=7m, t=3m $\\ $\frac{M}{N}=\frac{4}{7}:n=8,k=7m, t=4m$ } & $\approx 9.5$ &$O(\frac{1}{\sqrt{m}}119^m)$\\ \cline{3-5}

   &  &\tabincell{c}{Scheme in \cite{CJYT}\\ $k=28,t=13,n=2m-1$}& $15$ &$\frac{1}{28}784^{m}$\\ \cline{3-5}

    & &\tabincell{c}{Scheme in \cite{TR} and memory sharing\\ $\frac{M}{N}=\frac{1}{4}:k=4,l=14m,$ \\ $\ \ \ \ \ \ \ \ \ \ \ n=14m-1,x=1$\\$\frac{M}{N}=\frac{1}{2}:k=2,l=28m,$ \\ $\ \ \ \ \ \ \ \ \ \ \ \ \ \ \ \ n=28m-1,x=1$ } & $1.1429$ &$O(16384^{2m})$\\
    \hline

   \multirow{4}*{$165m$} & \multirow{4}*{$\frac{31}{55}$} &\tabincell{c}{Scheme A in Corollary \ref{strongbs}\\  $H=11,b=2,r=1$} & $9$ &$55^m$  \\ \cline{3-5}

    & &\tabincell{c}{Grouping method in \cite{SJTLD} and memory sharing \\$\frac{M}{N}=\frac{6}{11}:n=15,k=11m,t=6m$ \\$\frac{M}{N}=\frac{7}{11}:n=15,k=11m,t=7m$}  & $\approx 11.7143$ & $O(\frac{1}{\sqrt{m}}1957^m)$\\ \cline{3-5}

    & &\tabincell{c}{Scheme in \cite{CJYT}\\ $k=55,t=31,n=3m-1$}  & $12$ & $\frac{2}{55}55^{3m}$\\ \cline{3-5}

   &  &\tabincell{c}{Scheme in \cite{TR} and memory sharing\\ $\frac{M}{N}=\frac{1}{3}:k=3,l=55m,$ \\ $\ \ \ \ \ \ \ \ \ \ \ \ \ \ \ \ n=55m-1,x=1$ \\$\frac{M}{N}=\frac{37}{55}:k=3,l=55m,$ \\ $\ \ \ \ \ \ \ \ \ \ \ \ \ \ \ \ \ \ \ n=54m-1,x=54$} & $0.976$ &$O(177147^{5m})$\\

    \hline

  \end{tabular}}
\end{table}

\item {\bf Scheme B in Theorem \ref{multibs}.}
    \begin{itemize}
    \item {\it Comparison to the scheme in \cite{CJYT}.}
      For any positive integers $m,g$, by letting  $K=m(\lceil\frac{g^2}{2}\rceil+g)$ and $\frac{M}{N}=\frac{\lceil\frac{g^2}{2}\rceil}{\lceil\frac{g^2}{2}\rceil+g}$,
    let $k=\lceil\frac{g^2}{2}\rceil+g$, $t=\lceil\frac{g^2}{2}\rceil$ and $n=m-1$, we have $\lfloor\frac{k-1}{k-t}\rfloor=\lceil\frac{g}{2}\rceil$. Then the coded caching gain and subpacketization of the scheme in \cite{CJYT} listed in the third row of Table \ref{knownPDA} are
$$g_{CJ}=(n+1)\left\lfloor\frac{k-1}{k-t}\right\rfloor=m\left\lceil\frac{g}{2}\right\rceil \ \ \text{and} \ \ F_{CJ}=\left\lfloor\frac{k-1}{k-t}\right\rfloor k^{n}=\left\lceil\frac{g}{2}\right\rceil\left(\left\lceil\frac{g^2}{2}\right\rceil+g\right)^{m-1},$$ respectively. While the coded caching gain and subpacketization of Scheme B are
$$g_{Th3}=mg \ \  \text{and} \ \ F_{Th3}=\left(\left\lceil\frac{g^2}{2}\right\rceil+g\right)^m,$$ respectively. Then $$\frac{F_{Th3}}{F_{CJ}}=\frac{\lceil\frac{g^2}{2}\rceil+g}{\lceil\frac{g}{2}\rceil}.$$
In this case,  the coded caching gain of Scheme B is almost twice as much as that of the scheme in \cite{CJYT}, while their subpacketizations are of almost the  same order of magnitude.

 \item  {\it Numerical comparison.}
  The numerical comparison of Scheme B with the schemes in \cite{CJYT,SJTLD,TR} is presented in Table \ref{tablecom1}. It can be seen that Scheme B has a much lower subpacketization than the grouping method in \cite{SJTLD} while their loads are almost the same; Scheme B has a lower load than the scheme in \cite{CJYT} while their subpacketizations are of the same order of magnitude when $g$ is odd; Scheme B has a significant advantage in the subpacketization compared to the schemes in \cite{CJYT,TR}  when $g$ is even.

 \begin{table}
  \centering
  \caption{Numerical comparison  of Scheme B in Theorem \ref{multibs} with the schemes in \cite{CJYT,SJTLD,TR}  }\label{tablecom1}
  \small{
  \begin{tabular}{|c|c|c|c|c|}
   \hline
     $K$ &  $\frac{M}{N}$ &Scheme & $R$ &  $F$ \\ \hline

     \multirow{4}*{$24m$} & \multirow{4}*{$\frac{3}{4}$} &\tabincell{c}{Scheme B in Theorem \ref{multibs}\\  $g=6$} & $1$ &$24^m$  \\ \cline{3-5}

     & &\tabincell{c}{Grouping method in \cite{SJTLD} \\$n=3,k=8m,t=6m$ }  & $\approx 1$ & $O(\frac{1}{\sqrt{m}}90^m)$\\ \cline{3-5}

     & &\tabincell{c}{Scheme in \cite{CJYT}\\ $k=4,t=3,n=6m-1$}  & $0.3333$ & $\frac{3}{4}4096^{m}$\\ \cline{3-5}

     &  &\tabincell{c}{Scheme in \cite{TR}\\ $k=3,l=8m,n=6m-1,x=3$ } & $0.375$ &$\frac{8}{3}729^m$\\

    \hline

    \multirow{4}*{$32m$} & \multirow{4}*{$\frac{25}{32}$} &\tabincell{c}{Scheme B in Theorem \ref{multibs}\\ $g=7$} & $1$ &$32^m$  \\ \cline{3-5}

    & &\tabincell{c}{Grouping method in \cite{SJTLD} and memory sharing \\$\frac{M}{N}=\frac{6}{8}:n=4,k=8m,t=6m$\\ $\frac{M}{N}=\frac{7}{8}:n=4,k=8m,t=7m$} & $\approx  1$ & $O(\frac{1}{\sqrt{m}}90^m)$\\ \cline{3-5}

    & &\tabincell{c}{Scheme in \cite{CJYT}\\ $k=32,t=25,n=m-1$}  & $1.75$ & $\frac{1}{8}32^{m}$\\ \cline{3-5}

    & &\tabincell{c}{Scheme in \cite{TR} \\ $k=4,l=8m,n=7m-1,x=7$}  & $0.2917$ & $6*16384^{m}$\\

    \hline

  \end{tabular}}
\end{table}
\end{itemize}

\item {\bf Scheme C in Theorem \ref{MNbs}.}
\begin{itemize}
 \item {\it Comparison to the MN scheme in \cite{MN}.}
 For any   positive integers $q,z,m$ with $z<q$ and $m\geq 2$, by letting $K=mq$ and $\frac{M}{N}=\frac{z}{q}$, the coded caching gain, subpacketization, and load of the MN scheme are $$g_{MN}=mz+1, \ \ F_{MN}={mq\choose mz}, \ \ \text{and} \ \ R_{MN}=\frac{mq(1-\frac{z}{q})}{mz+1},$$ respectively. In this case, the coded caching gain, subpacketization, and load of Scheme C are
    $$g_{Th4}=mz, \ \ F_{{Th4}}=z{q\choose z}^m, \ \ \text{and} \ \ R_{{Th4}}=\frac{q-z}{z},$$ respectively.
    Hence, there is only one additive loss in the coded caching gain for Scheme C compared to the MN scheme. When $m\rightarrow \infty$, the ratio of the loads  $$\frac{R_{Th4}}{R_{MN}}=\frac{mz+1}{mz}$$ tends to $1$, so Scheme C is asymptotically optimal.
In addition, from  Stirling's Formula $n!\approx\sqrt{2\pi n}\left(\frac{n}{e}\right)^n$ (with $n\rightarrow \infty$), when $m\rightarrow \infty$  we have
\begin{eqnarray*}
F_{MN}=&\frac{(mq)!}{(mz)!(mq-mz)!} \ \ \ \ \ \ \ \ \ \ \ \ \ \ \ \ \ \ \ \ \ \ \ \ \ \ \ \ \ \ \ \ \    \\
\approx&\frac{\sqrt{2\pi(mq)}\left(\frac{mq}{e}\right)^{mq}}{\sqrt{2\pi(mz)}\left(\frac{mz}{e}\right)^{mz}\sqrt{2\pi(mq-mz)}\left(\frac{mq-mz}{e}\right)^{mq-mz}}\\
=&\sqrt{\frac{q}{2\pi z(q-z)m}}\left(\frac{q^q}{z^z(q-z)^{q-z}}\right)^m.  \ \ \ \ \ \ \ \ \ \ \ \ \ \ \ \ \
\end{eqnarray*}
Then the ratio of the subpacketizations is $$\frac{F_{Th4}}{F_{MN}}\approx z\sqrt{\frac{2\pi z(q-z)m}{q}}\left(\frac{q^q}{{q\choose z}z^z(q-z)^{q-z}}\right)^{-m}.$$
From binomial expansion, we have
\begin{eqnarray*}
q^q=&((q-z)+z)^q \ \ \ \ \ \ \ \ \ \ \ \ \ \ \ \ \ \ \ \ \ \ \ \ \ \ \ \ \ \ \ \ \ \ \ \ \ \ \ \ \ \ \ \ \ \ \ \ \ \ \ \ \ \ \ \ \ \ \ \ \ \ \ \ \ \ \\
   =&(q-z)^q+{q\choose 1}(q-z)^{q-1}z+\ldots+{q\choose z}(q-z)^{q-z}z^z+\ldots+z^q \ \ \ \ \ \ \ \ \ \\
   \geq & {q\choose z-1}(q-z)^{q-z+1}z^{z-1}+{q\choose z}(q-z)^{q-z}z^z+{q\choose z+1}(q-z)^{q-z-1}z^{z+1}.\\
\end{eqnarray*}
Since
\begin{eqnarray*}
&{q\choose z-1}(q-z)^{q-z+1}z^{z-1}+{q\choose z+1}(q-z)^{q-z-1}z^{z+1} \\
=&{q\choose z}(q-z)^{q-z}z^z\left(\frac{q-z}{q-z+1}+\frac{z}{z+1}\right) \ \ \ \ \ \ \ \ \ \ \ \ \ \ \ \ \ \ \
\end{eqnarray*}
and $$\frac{q-z}{q-z+1}+\frac{z}{z+1}=\frac{2(q-z)z+q}{(q-z)z+q+1}\geq 1,$$
we have $q^q\geq 2{q\choose z}(q-z)^{q-z}z^z$, i.e.,
$$\frac{q^q}{{q\choose z}z^z(q-z)^{q-z}}\geq 2.$$
Hence, we have
\begin{align*}
 & \frac{F_{Th4}}{F_{MN}}\approx z\sqrt{\frac{2\pi z(q-z)m}{q}}\left(\frac{q^q}{{q\choose z}z^z(q-z)^{q-z}}\right)^{-m} \\
 & \leq z\sqrt{\frac{2\pi z(q-z)m}{q}} 2^{-m} \\
 & =O\left(\sqrt{m}2^{-m}\right) \\
 & =O\left(\sqrt{\frac{K}{q}}2^{-\frac{K}{q}}\right).
\end{align*}

\item  {\it Comparison to the partition PDA scheme in \cite{YCTC}.}
For any positive integers $m,q$ with $q\geq 2$, when $K=mq$ and $\frac{N}{M}=\frac{1}{q}$,  the load and subpacketization of the partition PDA scheme are
 $$R_{YC}=q-1 \ \ \text{and} \ \ F_{YC}=q^{m-1},$$
 respectively. While the load and subpacketization of Scheme C are
 $$R_{Th4}=q-1 \ \ \text{and} \ \ F_{Th4}=q^{m},$$
 respectively. By letting $K=mq$ and $\frac{N}{M}=\frac{q-1}{q}$, the load and subpacketization of the partition PDA scheme are
 $$R_{YC}=\frac{1}{q-1} \ \ \text{and} \ \ F_{YC}=(q-1)q^{m-1}$$
 respectively. In this case, the load and subpacketization of Scheme C are
 $$R_{Th4}=\frac{1}{q-1} \ \ \text{and} \ \ F_{Th4}=(q-1)q^{m},$$
 respectively. Hence, when $K=mq$ and $\frac{M}{N}=\frac{1}{q}$ or $\frac{q-1}{q}$, Scheme C has the same load as the partition PDA scheme while its subpacketization is $q$ times of that of the partition PDA scheme. However, the partition PDA scheme can be used only when $\frac{M}{N}=\frac{1}{q}$ or $\frac{q-1}{q}$, while Scheme C can be used for any memory ratio satisfying $\frac{KM}{N}\in \mathbb{Z}$. For example, when $K=7843$, the memory-load and memory-subpacketization tradeoffs  are given in Fig. \ref{comth4R} and Fig. \ref{comth4F}, respectively.
 It can be seen that Scheme C has a significant advantage in subpacketization compared to the MN scheme and has a significant advantage in load compared to the partition PDA scheme.

\begin{figure}
\centering
\includegraphics[width=4in]{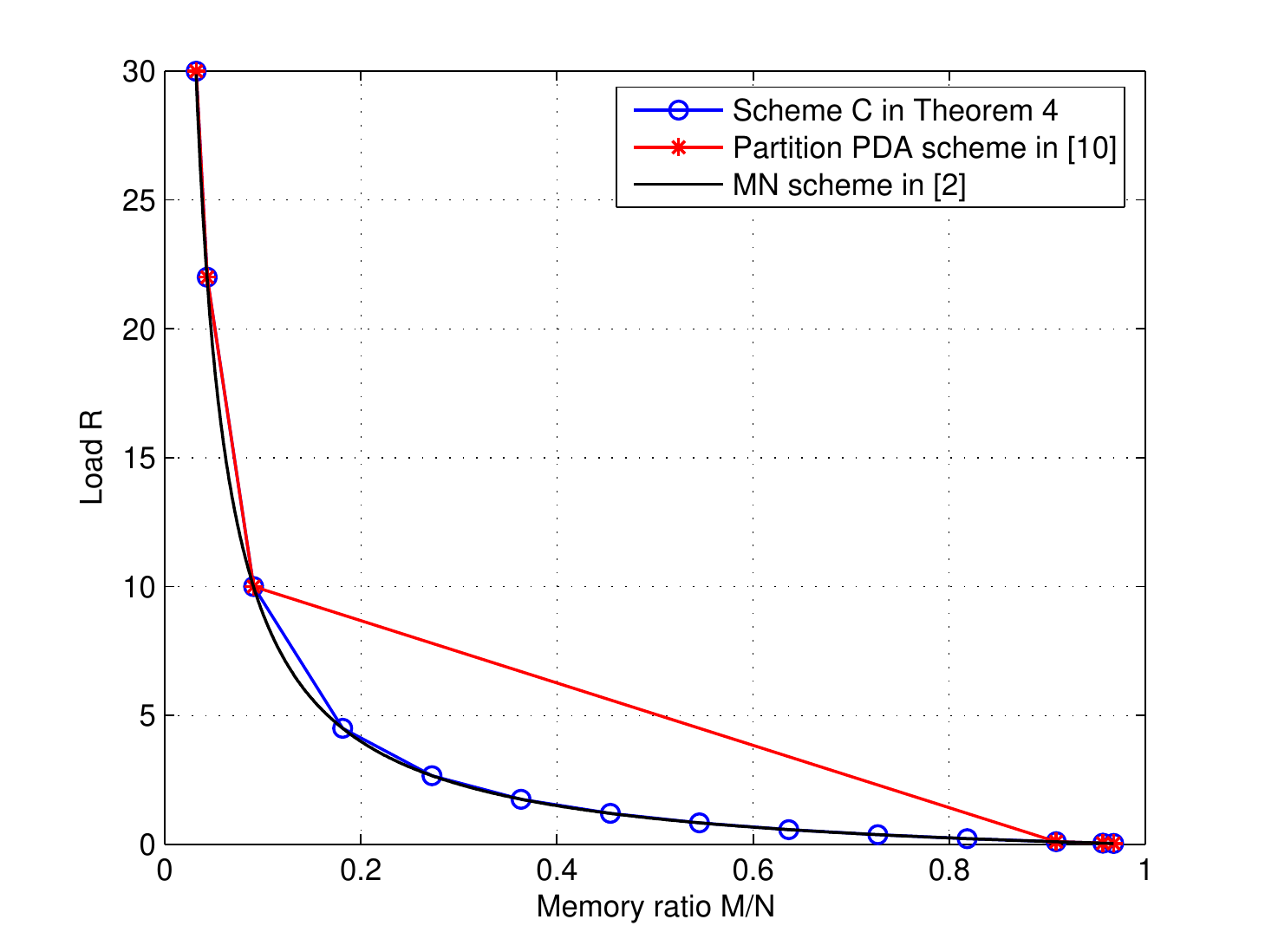}
\vskip 0.2cm
\caption{The load versus memory ratio for Scheme C in Theorem \ref{MNbs}, the partition PDA scheme in \cite{YCTC} and the MN scheme in \cite{MN} when $K=7843$. }\label{comth4R}
\end{figure}
\begin{figure}
\centering
\includegraphics[width=4in]{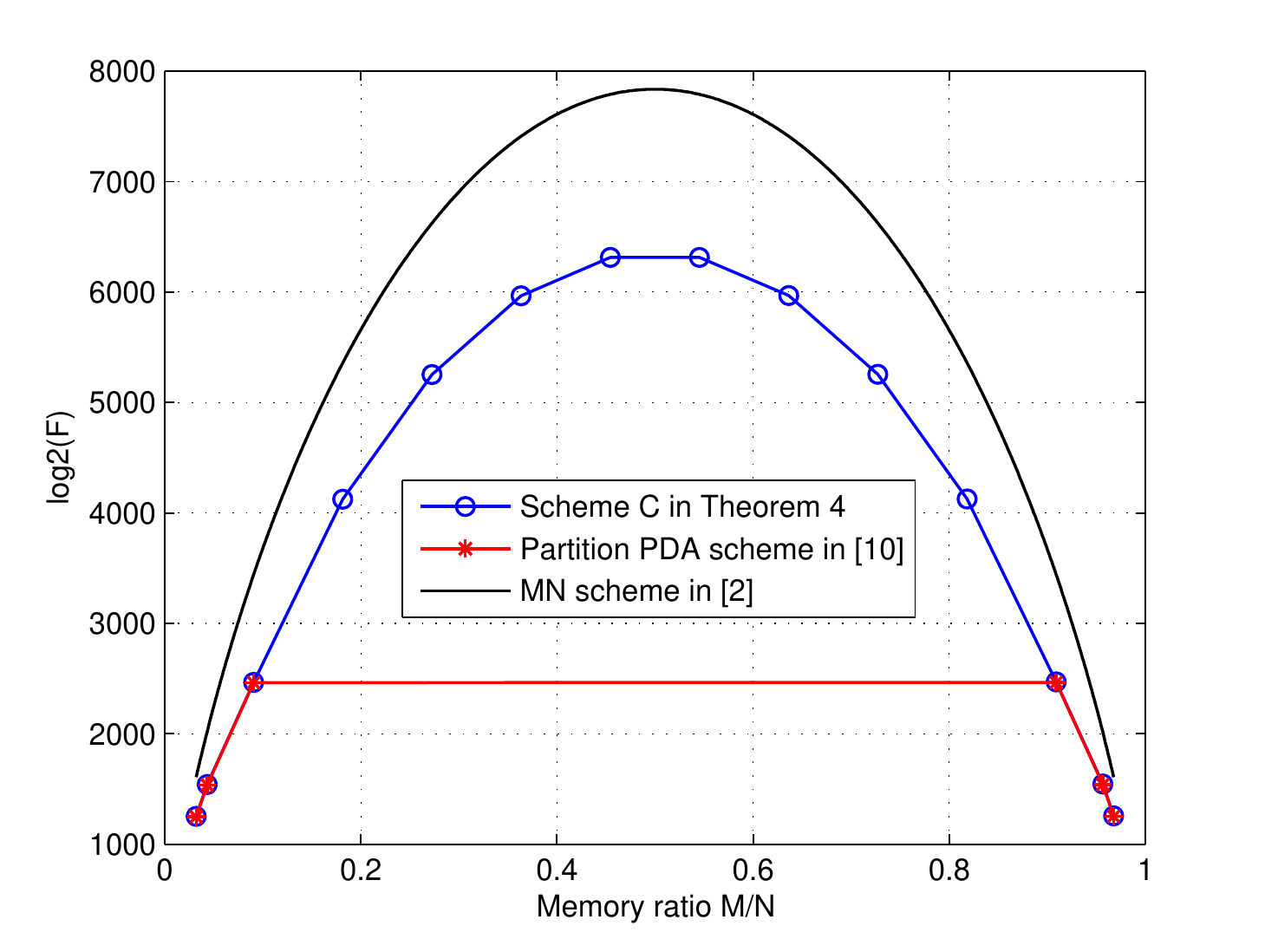}
\vskip 0.2cm
\caption{The base 2 logarithm of subpacketization versus memory ratio for Scheme C in Theorem \ref{MNbs}, the partition PDA scheme in \cite{YCTC} and the MN scheme in \cite{MN} when $K=7843$.   }\label{comth4F}
\end{figure}

\end{itemize}
\end{itemize}

\section{conclusion}
\label{conclusion}
In this paper, we studied the shared-link coded caching problem based on the PDA construction.  We first proposed two novel frameworks for constructing a PDA for $mK_1$ users via Cartesian product based on some existing PDA for $K_1$ users satisfying some constraints, where the resulting PDA achieves  a similar load as the original $K_1$-users  PDA with a significantly reduced subpacketization compared to the direction extension of the original $K_1$-users PDA to the   $mK_1$-users system.
As applications of the two frameworks, three new coded caching schemes were obtained, each of which has a significant advantage in the subpacketization and simultaneously has a lower or almost the same load, compared with   previously known coded caching schemes. Especially, for the third scheme which works for any number of users and any memory regime,   it has an exponential reduction on the subpacketization compared to  the MN scheme, while its coded caching gain is only one less than that of the MN scheme.

Further work includes the generalization of the proposed frameworks to heterogeneous caches by taking Cartesian product of different PDAs.

\appendices

\section{Proof of Lemma~\ref{lem:Theoerem 1 PDA}}
\label{prtheopm}
\begin{proof}
We will prove that the array $\mathbf{P}_m$ generated by Construction \ref{constr1} satisfies the definition of PDA (i.e., Definition \ref{def-PDA}).
\begin{itemize}
\item For any $(\delta,b)\in\mathcal{K}$ and any $i\in[1:\lambda]$, for column $(\delta,b)$ of $\mathbf{P}_m$, there are $\frac{Z_1}{\lambda}(\frac{F_1}{\lambda})^{m-1}$ stars in the rows indexed by $[(i-1)\frac{F_1}{\lambda}+1:i\frac{F_1}{\lambda}]^m$ from \eqref{constrPm}, since for each column of $\mathbf{P}_1$, there are $\frac{Z_1}{\lambda}$ stars in the rows indexed by $[(i-1)\frac{F_1}{\lambda}+1:i\frac{F_1}{\lambda}]$ from the first item of Condition
    \ref{proper1}. Hence, there are $Z=\lambda \frac{Z_1}{\lambda}(\frac{F_1}{\lambda})^{m-1}=Z_1(\frac{F_1}{\lambda})^{m-1}$ stars in column $(\delta,b)$ of $\mathbf{P}_m$ from \eqref{rowindex}. Condition C$1$ of Definition \ref{def-PDA} holds.

\item For any non-star entry in $\mathbf{P}_m$, say $\mathbf{P}_m({\bf f},(\delta,b))={\bf e}$, we will prove that for any $\delta'\in[1:m]$ with $\delta\neq\delta'$, there exist ${\bf f}'\in\mathcal{F}$ and $b'\in[1:K_1]$, such that $\mathbf{P}_m({\bf f}',(\delta',b'))={\bf e}$.
    From \eqref{constrPm} and \eqref{constre} we have
    $${\bf e}=(B_{<f_1>_{\frac{F_1}{\lambda}}}[\mu],\ldots,B_{<f_{\delta-1}>_{\frac{F_1}{\lambda}}}[\mu],\mathbf{P}_1(f_{\delta},b),B_{<f_{\delta+1}>_{\frac{F_1}{\lambda}}}[\mu],\ldots,B_{<f_{m}>_{\frac{F_1}{\lambda}}[\mu]}),$$
    where $\mathbf{P}_1(f_{\delta},b)=B_l[\mu]$.
    Let $s'=B_{<f_{\delta'}>_{\frac{F_1}{\lambda}}}[\mu]$, then $s'\in[1:S_1]$ and there exist $j\in[1:F_1]$ and $k\in[1:K_1]$, such that $\mathbf{P}_1(j,k)=s'$. Let $f'_{\delta'}=j$, $b'=k$. Assume that $j\in[(i-1)\frac{F_1}{\lambda}+1:i\frac{F_1}{\lambda}]$ with some $i\in[1:\lambda]$, let $f'_h=(i-1)\frac{F_1}{\lambda}+<f_h>_{\frac{F_1}{\lambda}}$ for any $h\in[1:m]\setminus \{\delta,\delta'\}$ and $f'_{\delta}=(i-1)\frac{F_1}{\lambda}+l$, then we have ${\bf f}'\in\mathcal{F}$ and $\mathbf{P}_m({\bf f}',(\delta',b'))={\bf e}$ from \eqref{constrPm} and \eqref{constre}.
    So the total number of different vectors in $\mathbf{P}_m$ is the number of different vectors in columns $(1,a)$ of $\mathbf{P}_m$ with all $a\in[1:K_1]$.

    For any $s\in[1:S_1]$, assume that $\mathbf{P}_1(j,k)=s$, $s=B_l[\mu]$ and $j\in[(i-1)\frac{F_1}{\lambda}+1:i\frac{F_1}{\lambda}]$ with some $i\in[1:\lambda]$. Let $f_1=j$ and $b=k$, for any $f_2,f_3,\ldots,f_{m}\in[(i-1)\frac{F_1}{\lambda}+1:i\frac{F_1}{\lambda}]$, we have $$\mathbf{P}_m({\bf f},(1,b))=(s,B_{<f_2>_{\frac{F_1}{\lambda}}}[\mu],\ldots,B_{<f_{m}>_{\frac{F_1}{\lambda}}}[\mu])$$
    from \eqref{constrPm} and \eqref{constre}. It implies that there are $(\frac{F_1}{\lambda})^{m-1}$ different vectors ${\bf e}$ appearing in $\mathbf{P}_m$ satisfying $e_1=s$. Hence, there are totally $S=S_1(\frac{F_1}{\lambda})^{m-1}$ different vectors in $\mathbf{P}_m$. Condition C$2$ of Definition \ref{def-PDA} holds.
\item For any two distinct entries $\mathbf{P}_m({\bf f},(\delta,b))$ and $\mathbf{P}_m({\bf f}',(\delta',b'))$, if $\mathbf{P}_m({\bf f},(\delta,b))=\mathbf{P}_m({\bf f}',(\delta',b'))={\bf e}\neq *$, we will prove that $\mathbf{P}_m({\bf f},(\delta',b'))=\mathbf{P}_m({\bf f}',(\delta,b))=*$. From \eqref{constrPm} and \eqref{constre} we have
    \begin{eqnarray}
    \label{eqc3}
    \begin{array}{ccc}
    {\bf e}&=&(B_{<f_1>_{\frac{F_1}{\lambda}}}[\mu],\ldots,B_{<f_{\delta-1}>_{\frac{F_1}{\lambda}}}[\mu],\mathbf{P}_1(f_{\delta},b),B_{<f_{\delta+1}>_{\frac{F_1}{\lambda}}}[\mu],\ldots,B_{<f_{m}>_{\frac{F_1}{\lambda}}}[\mu])\\
    &=&(B_{<f'_1>_{\frac{F_1}{\lambda}}}[\mu'],\ldots,B_{<f'_{\delta'-1}>_{\frac{F_1}{\lambda}}}[\mu'],\mathbf{P}_1(f'_{\delta'},b'),B_{<f'_{\delta'+1}>_{\frac{F_1}{\lambda}}}[\mu'],\ldots,B_{<f'_{m}>_{\frac{F_1}{\lambda}}}[\mu']),
    \end{array}
    \end{eqnarray}
    where $\mathbf{P}_1(f_\delta,b)=B_l[\mu]$ and $\mathbf{P}_1(f'_{\delta'},b')=B_{l'}[\mu']$.
    \begin{itemize}
    \item If $\delta=\delta'$, we have $<f_h>_{\frac{F_1}{\lambda}}=<f'_h>_{\frac{F_1}{\lambda}}$ for any $h\in[1:m]\setminus\{\delta\}$, $\mu=\mu'$ and $\mathbf{P}_1(f_\delta,b)=\mathbf{P}_1(f'_{\delta},b')\neq *$ from \eqref{eqc3}. Then $f_\delta=f'_\delta, b=b'$ or $\mathbf{P}_1(f_\delta,b')=\mathbf{P}_1(f'_{\delta},b)=*$ from Condition C$3$ of Definition \ref{def-PDA}. If $f_\delta=f'_\delta, b=b'$ hold, assume that $f_\delta\in[(i-1)\frac{F_1}{\lambda}+1:i\frac{F_1}{\lambda}]$ with some $i\in[1:\lambda]$, then from \eqref{rowindex} we have $f_h,f'_h\in[(i-1)\frac{F_1}{\lambda}+1:i\frac{F_1}{\lambda}]$ for any $h\in[1:m]$, which implies that $\lceil\frac{f_h}{F_1/\lambda}\rceil=\lceil\frac{f'_h}{F_1/\lambda}\rceil=i$. Hence, we have ${\bf f}={\bf f}'$, which contradicts the hypothesis that $\mathbf{P}_m({\bf f},(\delta,b))$ and $\mathbf{P}_m({\bf f}',(\delta,b'))$ are two distinct entries. So $\mathbf{P}_1(f_\delta,b')=\mathbf{P}_1(f'_{\delta},b)=*$ holds, which implies $\mathbf{P}_m({\bf f},(\delta,b'))=\mathbf{P}_m({\bf f}',(\delta,b))=*$ from \eqref{constrPm}.
    \item If $\delta\neq\delta'$, we have $\mathbf{P}_1(f_\delta,b)=B_{<f'_\delta>_{\frac{F_1}{\lambda}}}[\mu']$ and $\mathbf{P}_1(f'_{\delta'},b')=B_{<f_{\delta'}>_{\frac{F_1}{\lambda}}}[\mu]$ from \eqref{eqc3}, which implies that the $<f'_{\delta}>_{\frac{F_1}{\lambda}}$-th and $<f_{\delta'}>_{\frac{F_1}{\lambda}}$-th row of $\mathbf{P}_1$ are star rows for $\mathbf{P}_1(f_\delta,b)$ and $\mathbf{P}_1(f'_{\delta'},b')$ respectively, i.e., $\mathbf{P}_1(<f'_\delta>_{\frac{F_1}{\lambda}},b)=\mathbf{P}_1(<f_{\delta'}>_{\frac{F_1}{\lambda}},b')=*$, then we have $\mathbf{P}_1(f'_\delta,b)=\mathbf{P}_1(f_{\delta'},b')=*$ from the first item of Condition \ref{proper1}. Hence, we have $\mathbf{P}_m({\bf f}',(\delta,b))=\mathbf{P}_m({\bf f},(\delta',b'))=*$ from \eqref{constrPm}.
    \end{itemize}
    Condition C$3$ of Definition \ref{def-PDA} holds.
\end{itemize}
So $\mathbf{P}_m$ is an $(mK_1,\lambda (\frac{F_1}{\lambda})^m, Z_1(\frac{F_1}{\lambda})^{m-1},S_1(\frac{F_1}{\lambda})^{m-1})$ PDA.
\end{proof}

\section{Proof of Corollary \ref{cothm1}}
\label{prcothm1}
\begin{proof}
For the array $\mathbf{P}_m$ generated by Construction \ref{constr1}, from the proof of Theorem \ref{theopm}, we know that for any non-star entry, say $\mathbf{P}_m({\bf f},(\delta,b))={\bf e}\neq *$, for any $\delta'\in[1:m]$ with $\delta'\neq\delta$, there exist ${\bf f'}\in \mathcal{F}$ and $b'\in[1:K_1]$, such that $\mathbf{P}_m({\bf f'},(\delta',b'))={\bf e}$. So it is sufficient to prove that for any non-star entry in $\mathbf{P}_m$, say $\mathbf{P}_m({\bf f},(\delta,b))={\bf e}$, ${\bf e}$ appears $g$ times in columns $(\delta,a)$ of $\mathbf{P}_m$ with all $a\in[1:K_1]$.
Since $\mathbf{P}_m({\bf f},(\delta,b))={\bf e}$, we have $\mathbf{P}_1(f_{\delta},b)\neq *$ from \eqref{constrPm}. Let $\mathbf{P}_1(f_{\delta},b)=s$, then $s$ appears $g$ times in $\mathbf{P}_1$, since $\mathbf{P}_1$ is a $g$-PDA. Assume that $\mathbf{P}_1(j_{s,1},k_{s,1})=\mathbf{P}_1(j_{s,2},k_{s,2})=\ldots=\mathbf{P}_1(j_{s,g},k_{s,g})=s$. For any $h\in[1:g]$, assume that $j_{s,h}\in[(i_h-1)\frac{F_1}{\lambda}+1:i_h\frac{F_1}{\lambda}]$ with some $i_h\in[1:\lambda]$, then we have $\mathbf{P}_m(((i_h-1)F_1+<f_1>_{\frac{F_1}{\lambda}},\ldots,(i_h-1)F_1+<f_{\delta-1}>_{\frac{F_1}{\lambda}},j_{s,h},(i_h-1)F_1+<f_{\delta+1}>_{\frac{F_1}{\lambda}},\ldots,(i_h-1)F_1+<f_{m}>_{\frac{F_1}{\lambda}}),(\delta, k_{s,h}))={\bf e}$ from \eqref{constrPm} and \eqref{constre}. The proof is complete.
\end{proof}

\section{Proof of Proposition \ref{prostrong}}
\label{prprostrong}
\begin{proof}
The PDA in \cite{YTCC} is constructed as follows.
\begin{construction}(\cite{YTCC})
\label{constrct1}For any positive integers $H$, $a$, $b$, $r$ satisfying $\max\{a,b\}<H$, $r<\min\{a,b\}$ and $a+b\leq H+r$, let the row index set $\mathcal{F}={[1:H]\choose b}$ and the column index set $\mathcal{K}={[1:H]\choose a}$. There exists an ${H-a-b+2r\choose r}$-$\left({H\choose a},{H\choose b},{H\choose b}-{a\choose r}{H-a\choose b-r},{H\choose a+b-2r}{a+b-2r\choose a-r}\right)$ PDA $\mathbf{P}=(\mathbf{P}(B,A)|B\in \mathcal{F}, A\in \mathcal{K})$, defined as
\begin{eqnarray}
\label{eq-rule2}
\mathbf{P}(B,A)=\left\{
\begin{array}{cc}
((A\cup B)-(A\cap B),A-B)& \hbox{if} \ |A\cap B|=r,\\
*&\hbox{Otherwise}.
\end{array}
\right.
\end{eqnarray}
\end{construction}

When $a=b+r$, we will prove the PDA $\mathbf{P}$ generated by Construction \ref{constrct1} satisfies Condition \ref{proper1} with $\lambda=1$. Firstly, $\mathbf{P}$ satisfies the first item of Condition \ref{proper1} obviously. Secondly, for each non-star entry in $\mathbf{P}$, say $\mathbf{P}(B,A)=((A\cup B)-(A\cap B),A-B)$, it implies that $|A\cap B|=r$ and $|A-B|=a-r=b$. Hence, we have $A-B\in \mathcal{F}$. Moreover, $\mathbf{P}(A-B,A)=*$ since $|(A-B)\cap A|=|A-B|=b>r$. So the row $A-B$ is a star row for the non-star entry $((A\cup B)-(A\cap B),A-B)$. Let $\phi: ((A\cup B)-(A\cap B),A-B)\mapsto A-B$, then for each row $B'\in \mathcal{F}$, there are ${H-b\choose b-r}$ different non-star entries for which the assigned star row is the row $B'$, since the number of sets $(A\cup B)-(A\cap B)=(A-B)\cup(B-A)$ satisfying $A-B=B'$ is ${H-b\choose b-r}$. $\mathbf{P}$ satisfies the second item of Condition \ref{proper1}. So $\mathbf{P}$ satisfies Condition \ref{proper1} with $\lambda=1$. The proof is complete.
\end{proof}

\section{Proof of Proposition \ref{proPDAOA}}
\label{prproPDAOA}
\begin{proof}
The ${m\choose t}$-$\left({m\choose t}q^t,q^{m-1}, q^{m-1}-(q-1)^tq^{m-t-1}, (q-1)^tq^{m-1}\right)$ PDA in \cite{CWZW} is constructed as follows.
\begin{construction}(\cite{CWZW})
\label{conPDAOA}
For any positive integers $m,q,t$ with $t< m$ and $q\geq 2$, let the column index set $\mathcal{K}={[1:m]\choose t}\times [1:q]^t$ and the row index set
$$\mathcal{F}=\left\{\left(f_{1},f_{2},\ldots,f_{m-1}, <\sum_{i=1}^{m-1} f_i>_q\right)\ |\  f_1,f_2,\ldots,f_{m-1}\in [1:q]\right\},$$
There exists an ${m\choose t}$-$\left({m\choose t}q^t,q^{m-1}, q^{m-1}-(q-1)^tq^{m-t-1}, (q-1)^tq^{m-1}\right)$ PDA $\mathbf{P}=(\mathbf{P}({\bf f},{\bf k}))$ with ${\bf f}=(f_1,f_2,\ldots, f_{m})\in \mathcal{F}$ and  ${\bf k}=({\mathcal{T}},{\bf b})=(\{\delta_1,\delta_2,\ldots, \delta_{t}\},(b_1,b_2,\ldots,b_{t}))\in \mathcal{K}$ satisfying $1\leq\delta_1<\delta_2<\ldots<\delta_{t}\leq m$, defined as
\begin{eqnarray}
\label{eq-constr-PDA}
\mathbf{P}({\bf f},{\bf k})=\left\{
\begin{array}{ll}
({\bf e},n_{\bf e}) & \textrm{if}~d({\bf f}[\mathcal{T}]$, ${\bf b})=t, \\
* & \textrm{otherwise},
\end{array}
\right.
\end{eqnarray}
where ${\bf e}=(e_1,e_{2},\ldots,e_{m})\in[1:q]^m$ such that  \begin{eqnarray}
\label{eq-putting integer}
e_i=\left\{
\begin{array}{ll}
b_h & \textrm{if}\ i=\delta_h, h\in [1:t],\\[0.2cm]
f_i & \textrm{otherwise} \end{array}
\right.
\end{eqnarray}
and $n_{\bf e}$ is the  occurrence  order of ${\bf e}$ occurring in column ${\bf k}$. If $({\bf e},n_{\bf e})$ occurs in $\mathbf{P}$, we also say ${\bf e}$ occurs in $\mathbf{P}$.
\end{construction}

\begin{proposition}
\label{prooccurrence}
For the array $\mathbf{P}$ generated by Construction \ref{conPDAOA}, when $t\geq 2$, for each vector ${\bf e}\in[1:q]^m$, if ${\bf e}\in \mathcal{F}$, then ${\bf e}$ appears $(q-1)^{t-1}-(q-1)^{t-2}+\ldots+(-1)^t(q-1)$ times in each column where it appears; if ${\bf e}\notin \mathcal{F}$, then ${\bf e}$ appears $(q-1)^{t-1}-(q-1)^{t-2}+\ldots+(-1)^{t+1}$ times in each column where it appears.
\end{proposition}
\begin{proof}
Here mathematical induction will be used. Firstly, Lemma 1 in \cite{CWZW} showed that for any non-star entry $({\bf e}, n_{\bf e})$ appearing in $\mathbf{P}$, the occurrence number of ${\bf e}$ in each column where it appears is the same.
\begin{itemize}
\item When $t=2$, for each vector ${\bf e}\in[1:q]^m$, if ${\bf e}\in \mathcal{F}$, without loss of generality, assume that ${\bf e}=(q,q,\ldots,q)$, let us consider the occurrence number of ${\bf e}$ in each column where it appears, such as column ${\bf k}=(\{1,2\},(q,q))$. From Construction \ref{conPDAOA} we have
    \begin{eqnarray*}
    \mathbf{P}((1,q-1,q,\ldots,q),{\bf k})&=&({\bf e},1), \\
    \mathbf{P}((2,q-2,q,\ldots,q),{\bf k})&=&({\bf e},2),\\
     &\vdots&\\
     \mathbf{P}((q-1,1,q,\ldots,q),{\bf k})&=&({\bf e},q-1).
    \end{eqnarray*}
    So ${\bf e}$ occurs $q-1$ times in column ${\bf k}$.
    If ${\bf e}\notin \mathcal{F}$, without loss of generality, assume that ${\bf e}=(q,\ldots,q,1)$, let us consider the occurrence number of ${\bf e}$ in each column where it appears, such as column ${\bf k}=(\{1,2\},(q,q))$. From Construction \ref{conPDAOA} we have
    \begin{eqnarray*}
    \mathbf{P}((2,q-1,q,\ldots,q,1),{\bf k})&=&({\bf e},1),\\
     \mathbf{P}((3,q-2,q,\ldots,q,1),{\bf k})&=&({\bf e},2),\\
      &\vdots& \\
      \mathbf{P}((q-1,2,q,\ldots,q,1),{\bf k})&=&({\bf e},q-2).
    \end{eqnarray*}
     So ${\bf e}$ occurs $q-2=(q-1)-1$ times in column ${\bf k}$. The proposition holds when $t=2$.

\item Assume that the proposition holds when $t=n$, i.e., for each vector ${\bf e}\in[1:q]^m$, if ${\bf e}\in \mathcal{F}$, assume that ${\bf e}=(q,q,\ldots,q)$, ${\bf e}$ appears $(q-1)^{n-1}-(q-1)^{n-2}+\ldots+(-1)^n(q-1)$ times in column ${\bf k}=(\{1,2,\ldots,n\},(q,q,\ldots,q))$, which implies that the number of vectors $(f_1,\ldots,f_{n},q,\ldots,q)$ satisfying $f_i\in[1:q-1]$ for any $i\in[1:n]$ and $<\sum_{i=1}^{n}f_i>_q=q$ is $(q-1)^{n-1}-(q-1)^{n-2}+\ldots+(-1)^n(q-1)$. If ${\bf e}\notin \mathcal{F}$, assume that ${\bf e}=(q,q,\ldots,q,\alpha)$ where $\alpha\in[1:q-1]$, ${\bf e}$ appears $(q-1)^{n-1}-(q-1)^{n-2}+\ldots+(-1)^{n+1}$ times in column ${\bf k}=(\{1,2,\ldots,n\},(q,q,\ldots,q))$, which implies that the number of vectors $(f_1,\ldots,f_{n},q,\ldots,q,\alpha)$ satisfying $f_i\in[1:q-1]$ for any $i\in[1:n]$ and $<\sum_{i=1}^{n}f_i>_q=\alpha$ is $(q-1)^{n-1}-(q-1)^{n-2}+\ldots+(-1)^{n+1}$.

 \item When $t=n+1$, for each vector ${\bf e}\in[1:q]^m$, if ${\bf e}\in \mathcal{F}$, assume that ${\bf e}=(q,q,\ldots,q)$, the occurrence number of ${\bf e}$ in column ${\bf k}=(\{1,2,\ldots,n+1\},(q,q,\ldots,q))$ is equal to the number of vectors $(f_1,\ldots,f_{n+1},q,\ldots,q)$ satisfying $f_i\in[1:q-1]$ for any $i\in[1:n+1]$ and $<\sum_{i=1}^{n+1}f_i>_q=q$.
     For any $\beta\in[1:q-1]$, when $f_{n+1}=\beta$, the number of vectors $(f_1,\ldots,f_{n},\beta,q,\ldots,q)$ with $f_i\in[1:q-1]$ for any $i\in[1:n]$ and $<\sum_{i=1}^{n+1}f_i>_q=q$ is the number of vectors $(f_1,\ldots,f_{n},\beta,q,\ldots,q)$ satisfying $f_i\in[1:q-1]$ for any $i\in[1:n]$ and $<\sum_{i=1}^{n}f_i>_q=q-\beta$, which is $(q-1)^{n-1}-(q-1)^{n-2}+\ldots+(-1)^{n+1}$ by hypothesis. Hence, the occurrence number of ${\bf e}$ in column ${\bf k}$ is $(q-1)((q-1)^{n-1}-(q-1)^{n-2}+\ldots+(-1)^{n+1})=(q-1)^{n}-(q-1)^{n-1}+\ldots+(-1)^{n+1}(q-1)$.

     If ${\bf e}\notin \mathcal{F}$, assume that ${\bf e}=(q,\ldots,q,1)$, the occurrence number of ${\bf e}$ in column ${\bf k}=(\{1,2,\ldots,n+1\},(q,q,\ldots,q))$ is equal to the number of vectors $(f_1,\ldots,f_{n+1},q,\ldots,q,1)$ satisfying $f_i\in[1:q-1]$ for any $i\in[1:n+1]$ and $<\sum_{i=1}^{n+1}f_i>_q=1$.
     When $f_{n+1}=1$, the number of vectors $(f_1,\ldots,f_{n},1,q,\ldots,q,1)$ satisfying $f_i\in[1:q-1]$ for any $i\in[1:n]$ and $<\sum_{i=1}^{n+1}f_i>_q=1$ is equal to the number of vectors $(f_1,\ldots,f_{n},1,q,\ldots,q,1)$ satisfying $f_i\in[1:q-1]$ for any $i\in[1:n]$ and $<\sum_{i=1}^{n}f_i>_q=q$, which is $(q-1)^{n-1}-(q-1)^{n-2}+\ldots+(-1)^n(q-1)$ by hypothesis. For any $\beta\in[2:q-1]$, when $f_{n+1}=\beta$, the number of vectors $(f_1,\ldots,f_{n},\beta,q,\ldots,q,1)$ satisfying $f_i\in[1:q-1]$ for any $i\in[1:n]$ and $<\sum_{i=1}^{n+1}f_i>_q=1$ is equal to the number of vectors $(f_1,\ldots,f_{n},\beta,q,\ldots,q,1)$ with $f_i\in[1:q-1]$ for any $i\in[1:n]$ and $<\sum_{i=1}^{n}f_i>_q=q-\beta+1$, which is $(q-1)^{n-1}-(q-1)^{n-2}+\ldots+(-1)^{n+1}$ by hypothesis. Hence the occurrence number of ${\bf e}$ in column ${\bf k}$ is $(q-1)^{n-1}-(q-1)^{n-2}+\ldots+(-1)^n(q-1)+(q-2)((q-1)^{n-1}-(q-1)^{n-2}+\ldots+(-1)^{n+1})=(q-1)^{n}-(q-1)^{n-1}+\ldots+(-1)^{n+2}$. The proposition holds when $t=n+1$.
\end{itemize}
\end{proof}

Next we will prove that the PDA $\mathbf{P}$ generated by Construction \ref{conPDAOA} satisfies Condition \ref{proper1} with $\lambda=1$ when $t\geq 2$. Firstly, $\mathbf{P}$ satisfies the first item of Condition \ref{proper1} obviously.
Secondly, for any non-star entry $({\bf e}, n_{\bf e})$ in $\mathbf{P}$, define a mapping
\begin{eqnarray}
\label{fai}
\phi(({\bf e}, n_{\bf e}))=\left\{
\begin{array}{ll}
{\bf e} & \textrm{if}~{\bf e}\in \mathcal{F}, \\
\left(e_1,e_2,\ldots,e_{m-1},<\sum_{i=1}^{m-1}e_i>_q\right) & \textrm{otherwise},
\end{array}
\right.
\end{eqnarray}
then $\phi(({\bf e}, n_{\bf e}))\in\mathcal{F}$.
Moreover, if $\mathbf{P}({\bf f},(\mathcal{T},{\bf b}))=({\bf e}, n_{\bf e})$, we have ${\bf e}[\mathcal{T}]={\bf b}$ from \eqref{eq-putting integer}. If ${\bf e}\in \mathcal{F}$, we have $d(\phi(({\bf e}, n_{\bf e}))[\mathcal{T}],{\bf b})=d({\bf e}[\mathcal{T}],{\bf b})=0<t$, then $\mathbf{P}(\phi(({\bf e}, n_{\bf e})),(\mathcal{T},{\bf b}))=*$ from \eqref{eq-constr-PDA}. If ${\bf e}\notin \mathcal{F}$, we have $d(\phi(({\bf e}, n_{\bf e}))[\mathcal{T}],{\bf b})=d((e_1,e_2,\ldots,e_{m-1},<\sum_{i=1}^{m-1}e_i>_q)[\mathcal{T}],{\bf b})\leq 1<t$, then  $\mathbf{P}(\phi(({\bf e}, n_{\bf e})),(\mathcal{T},{\bf b}))=*$ from \eqref{eq-constr-PDA}. Hence, the row $\phi(({\bf e}, n_{\bf e}))$ is a star row for $({\bf e}, n_{\bf e})$.
For any ${\bf f}=(f_1,f_2,\ldots,f_{m})\in\mathcal{F}$, $B_{\bf f}$ denotes the set of non-star entries $({\bf e}, n_{\bf e})$ in $\mathbf{P}$ satisfying $\phi(({\bf e}, n_{\bf e}))={\bf f}$.
Next we will prove that $|B_{\bf f}|=(q-1)^t$ for each ${\bf f}\in\mathcal{F}$. From \eqref{fai}, $B_{\bf f}$ is a collection of $({\bf e}, n_{\bf e})$ with ${\bf e}={\bf f}$ or ${\bf e}=(f_1,\ldots,f_{m-1},\beta)$ with $\beta\neq f_{m}$. If ${\bf e}={\bf f}$, then ${\bf e}\in \mathcal{F}$, so ${\bf e}$ appears $(q-1)^{t-1}-(q-1)^{t-2}+\ldots+(-1)^t(q-1)$ times in each column where it appears from Proposition \ref{prooccurrence}. If ${\bf e}=(f_1,\ldots,f_{m-1},\beta)$ with $\beta\neq f_{m}$, then ${\bf e}\notin \mathcal{F}$, so ${\bf e}$ appears $(q-1)^{t-1}-(q-1)^{t-2}+\ldots+(-1)^{t+1}$ times in each column where it appears from Proposition \ref{prooccurrence}. Hence, $|B_{\bf f}|=(q-1)^{t-1}-(q-1)^{t-2}+\ldots+(-1)^t(q-1)+(q-1)\left((q-1)^{t-1}-(q-1)^{t-2}+\ldots+(-1)^{t+1}\right)=(q-1)^t$. $\mathbf{P}$ satisfies the second item of Condition \ref{proper1}.
Hence, when $t\geq 2$, the array $\mathbf{P}$ generated by Construction \ref{conPDAOA} satisfies Condition \ref{proper1} with $\lambda=1$.
\end{proof}

\section{Proof of Theorem \ref{multibs}}
\label{prmultibs}
In order to prove Theorem \ref{multibs}, we first construct a PDA satisfying Condition \ref{proper1} with $\lambda=1$.
\begin{construction}
\label{contru2}
For any positive integer $g$, let $q=\lceil\frac{g^2}{2}\rceil+g$ and $z=\lceil\frac{g^2}{2}\rceil$. Then a $q\times q$ array $\mathbf{P}$ is defined as
\begin{equation}
\label{constru2e}
\mathbf{P}(j,k)=\begin{cases}
* &\text{if} \ \ j\in[k-z+1:k]_q,\\
<j-(h-1)(g+2)>_q& \text{if} \ \ j=<k+h>_q,1\leq h\leq\lceil\frac{g}{2}\rceil, \\
<k-(g-h)(g+2)>_q& \text{if} \ \ j=<k+h>_q, \lceil\frac{g}{2}\rceil<h\leq g. \\
\end{cases}
\end{equation}
\end{construction}
\begin{example}
\label{exm2}
If $g=2$, then $q=\lceil\frac{g^2}{2}\rceil+g=4$ and $z=\lceil\frac{g^2}{2}\rceil=2$. The $4\times4$ array defined in \eqref{constru2e} is as follows, which is exactly the PDA in \eqref{pda2},
\begin{eqnarray*}
\mathbf{P}=\left(\begin{array}{cccc}
*&*&3&1\\
2&*&*&4\\
1&3&*&*\\
*&2&4&*
\end{array}\right).
\end{eqnarray*}
Let us consider the first column. When $k=1$, if $j\in[k-z+1:k]_q=[1-2+1:1]_4=\{<0>_4,<1>_4\}=\{4,1\}$, we have $\mathbf{P}(j,1)=*$ from \eqref{constru2e}; if $j=2$, from $j=<k+h>_q$ we have $h=1=\lceil\frac{g}{2}\rceil$, so $\mathbf{P}(2,1)=<j-(h-1)(g+2)>_q=<2>_4=2$; if $j=3$, from $j=<k+h>_q$ we have $h=2>\lceil\frac{g}{2}\rceil$, so $\mathbf{P}(3,1)=<k-(g-h)(g+2)>_q=<1>_4=1$. The other columns can be obtained similarly.
\hfill $\square$
\end{example}
\begin{lemma}
\label{lemma1}
The array $\mathbf{P}$ generated by Construction \ref{contru2} is a $g$-$(\lceil\frac{g^2}{2}\rceil+g,\lceil\frac{g^2}{2}\rceil+g,\lceil\frac{g^2}{2}\rceil,\lceil\frac{g^2}{2}\rceil+g)$ PDA satisfying Condition \ref{proper1} with $\lambda=1$.
\hfill $\square$
\end{lemma}

\begin{proof}
Firstly, we will prove that the array $\mathbf{P}$ generated by Construction \ref{contru2} is a $g$-$(\lceil\frac{g^2}{2}\rceil+g,\lceil\frac{g^2}{2}\rceil+g,\lceil\frac{g^2}{2}\rceil,\lceil\frac{g^2}{2}\rceil+g)$ PDA.
\begin{itemize}
\item Since $\mathbf{P}(j,k)=*$ if $j\in[k-z+1:k]_q$ from \eqref{constru2e}, there are $z=\lceil\frac{g^2}{2}\rceil$ stars in each column. Condition C$1$ of Definition \ref{def-PDA} holds.
\item For any $s\in[1:\lceil\frac{g^2}{2}\rceil+g]$, we have $\mathbf{P}(s,<s-1>_q)=s$ from \eqref{constru2e}. Condition C$2$ of Definition \ref{def-PDA} holds.
\item For any two distinct entries $\mathbf{P}(j_1,k_1)$ and $\mathbf{P}(j_2,k_2)$, if $\mathbf{P}(j_1,k_1)=\mathbf{P}(j_2,k_2)=s\in[1:\lceil\frac{g^2}{2}\rceil+g]$, we will prove $\mathbf{P}(j_1,k_2)=\mathbf{P}(j_2,k_1)=*$. Let $j_1=<k_1+h_1>_q$ and $j_2=<k_2+h_2>_q$. Since     \begin{equation}
    \label{eqstar}
    \mathbf{P}(j,k)=* \ \ \text{if} \ \ j\in\left[k-\left\lceil\frac{g^2}{2}\right\rceil+1:k\right]_q=\left[k+g+1:k+\left\lceil\frac{g^2}{2}\right\rceil+g\right]_q
    \end{equation}
    from \eqref{constru2e}, we have $h_1,h_2\in[1:g]$.
    \begin{itemize}
    \item If $h_1,h_2\in[1:\lceil\frac{g}{2}\rceil]$, then $s=<j_1-(h_1-1)(g+2)>_q=<j_2-(h_2-1)(g+2)>_q$ from \eqref{constru2e}. Consequently, we have  $<j_1-(h_1-1)(g+2)>_q=<k_2+h_2-(h_2-1)(g+2)>_q$, then
        \begin{equation}
        \label{j1k2}
        j_1=<k_2+h_1(g+2)-h_2(g+1)>_q.
        \end{equation}
        In addition, we have $h_1\neq h_2$. Otherwise if $h_1=h_2$, we have $j_1=j_2$ and $k_1=k_2$, which contradicts the hypothesis that $\mathbf{P}(j_1,k_1)$ and $\mathbf{P}(j_2,k_2)$ are two distinct entries.
        \begin{itemize}
        \item[$\diamond$] If $h_1>h_2$, then
        \begin{align*}
        h_1(g+2)-h_2(g+1)&\geq (h_2+1)(g+2)-h_2(g+1) \\& = g+2+h_2  \\&> g+1.
        \end{align*}
        On the other hand,
       \begin{align*}
       h_1(g+2)-h_2(g+1)&\leq \left\lceil\frac{g}{2}\right\rceil(g+2)-(g+1) \\ & \leq \frac{g+1}{2}(g+2)-(g+1) \\&=\frac{g^2}{2}+\frac{g}{2}\\& < \left\lceil\frac{g^2}{2}\right\rceil+g.
       \end{align*}
       So we have $\mathbf{P}(j_1,k_2)=*$ from \eqref{eqstar} and \eqref{j1k2}.
    \item[$\diamond$] If $h_1<h_2$, then
        \begin{align*}
       h_2(g+1)-h_1(g+2)& \leq \left\lceil\frac{g}{2}\right\rceil(g+1)-(g+2) \\& \leq \frac{(g+1)^2}{2}-(g+2)\\&=\frac{g^2}{2}-\frac{3}{2} \\& <\left\lceil\frac{g^2}{2}\right\rceil.
        \end{align*}
        On the other hand,
        \begin{align*}
        h_2(g+1)-h_1(g+2)&\geq (h_1+1)(g+1)-h_1(g+2) \\& =g+1-h_1 \\&>0.
        \end{align*}
        So we have $\mathbf{P}(j_1,k_2)=*$ from \eqref{eqstar} and \eqref{j1k2}.
        \end{itemize}
    \item If $h_1,h_2\in[\lceil\frac{g}{2}\rceil+1:g]$, then $s=<k_1-(g-h_1)(g+2)>_q=<k_2-(g-h_2)(g+2)>_q$ from \eqref{constru2e}. Consequently, we have  $<j_1-h_1-(g-h_1)(g+2)>_q=<k_2-(g-h_2)(g+2)>_q$, then
        \begin{equation}
        \label{j1k21}
        j_1=<k_2+h_2(g+2)-h_1(g+1)>_q.
        \end{equation}
        In addition, we also have $h_1\neq h_2$.
        \begin{itemize}
        \item[$\diamond$] If $h_1>h_2$, then
             \begin{align*}
                      h_1(g+1)-h_2(g+2)&\leq g(g+1)-\left(\left\lceil\frac{g}{2}\right\rceil+1\right)(g+2)\\& \leq g(g+1)-\left(\frac{g}{2}+1\right)(g+2)\\&=\frac{g^2}{2}-g-2 \\&< \left\lceil\frac{g^2}{2}\right\rceil.
             \end{align*}
        On the other hand,
        \begin{align*}
        h_1(g+1)-h_2(g+2)&\geq (h_2+1)(g+1)-h_2(g+2) \\&= g+1-h_2 \\&>0.
        \end{align*}
         So we have $\mathbf{P}(j_1,k_2)=*$ from \eqref{eqstar} and \eqref{j1k21}.
        \item[$\diamond$] If $h_1<h_2$, then
        \begin{align*}
         h_2(g+2)-h_1(g+1)&\leq g(g+2)-\left(\left\lceil\frac{g}{2}\right\rceil+1\right)(g+1) \\ & \leq g(g+2)-\left(\frac{g}{2}+1\right)(g+1) \\& =\frac{g^2}{2}+\frac{g}{2}-1 \\& <\left\lceil\frac{g^2}{2}\right\rceil+g.
        \end{align*}
        On the other hand,
        \begin{align*}
       h_2(g+2)-h_1(g+1)& \geq (h_1+1)(g+2)-h_1(g+1)\\&=h_1+g+2\\&>g+1.
        \end{align*}
        So we have $\mathbf{P}(j_1,k_2)=*$ from \eqref{eqstar} and \eqref{j1k21}.
        \end{itemize}

    \item If $h_1\in[1:\lceil\frac{g}{2}\rceil]$ and $h_2\in[\lceil\frac{g}{2}\rceil+1:g]$, then $s=<j_1-(h_1-1)(g+2)>_q=<k_2-(g-h_2)(g+2)>_q$ from \eqref{constru2e}. Consequently,
        \begin{equation}
        \label{j1k23}
        j_1=<k_2+(h_1+h_2-g-1)(g+2)>_q.
        \end{equation}
        \begin{itemize}
        \item[$\diamond$] If $h_1+h_2> g+1$, we have $h_1+h_2-g-1\geq1$, then
        \begin{align*}
         g+2&\leq(h_1+h_2-g-1)(g+2)\leq \left(\left\lceil\frac{g}{2}\right\rceil+g-g-1\right)(g+2) \\& \leq \left(\frac{g+1}{2}-1\right)(g+2)=\frac{g^2}{2}+\frac{g}{2}-1 \\& <\left\lceil\frac{g^2}{2}\right\rceil+g.
        \end{align*}
        So we have $\mathbf{P}(j_1,k_2)=*$ from \eqref{eqstar} and \eqref{j1k23}.
        \item[$\diamond$] If $h_1+h_2\leq g+1$, then
        \begin{align*}
        0&\leq(g+1-h_1-h_2)(g+2) \\& \leq \left(g+1-1-\left\lceil\frac{g}{2}\right\rceil-1\right)(g+2) \\& \leq \left(\frac{g}{2}-1\right)(g+2)=\frac{g^2}{2}-2 \\&<\left\lceil\frac{g^2}{2}\right\rceil.
        \end{align*}
        So we have $\mathbf{P}(j_1,k_2)=*$ from \eqref{eqstar} and \eqref{j1k23}.
        \end{itemize}

    \end{itemize}
    Similarly, we can prove that $\mathbf{P}(j_2,k_1)=*$. Condition C3 of Definition \ref{def-PDA} holds.
\end{itemize}

Next we will prove that for any $s\in[1:\lceil\frac{g^2}{2}\rceil+g]$, $s$ appears exactly $g$ times in $\mathbf{P}$. For any $h\in[1:g]$, if $h\leq \lceil\frac{g}{2}\rceil$, we have $$\mathbf{P}(<(h-1)(g+2)+s>_q, <(h-1)(g+2)+s-h>_q)=s$$ from \eqref{constru2e}; if $h> \lceil\frac{g}{2}\rceil$, we have $$\mathbf{P}(<(g-h)(g+2)+s+h>_q, <(g-h)(g+2)+s>_q)=s$$ from \eqref{constru2e}. So $s$ appears exactly $g$ times in $\mathbf{P}$.
Hence, $\mathbf{P}$ is a $g$-$(\lceil\frac{g^2}{2}\rceil+g,\lceil\frac{g^2}{2}\rceil+g,\lceil\frac{g^2}{2}\rceil,\lceil\frac{g^2}{2}\rceil+g)$ PDA.

Secondly, we will prove that $\mathbf{P}$ satisfies Condition \ref{proper1} with $\lambda=1$. $\mathbf{P}$ satisfies the first item of Condition \ref{proper1} obviously. For each $s\in[1:\lceil\frac{g^2}{2}\rceil+g]$, let $\phi(s)=<s-1>_q$, we will prove that the $\phi(s)$-th row is a star row for $s$, i.e., for any $j,k\in[1:q]$, if $\mathbf{P}(j,k)=s$, then $\mathbf{P}(\phi(s),k)=*$. Let $j=<k+h>_q$, then $h\in[1:g]$.
\begin{itemize}
\item If $h\in[1:\lceil\frac{g}{2}\rceil]$, then $s=<j-(h-1)(g+2)>_q=<k+h-(h-1)(g+2)>_q$ from \eqref{constru2e}. Consequently, $$\phi(s)=<s-1>_q=<k-(g+1)(h-1)>_q.$$ Since
\begin{align*}
0&\leq (g+1)(h-1) \\& \leq (g+1)\left(\left\lceil\frac{g}{2}\right\rceil-1\right) \\& \leq (g+1)\left(\frac{g+1}{2}-1\right) \\& =\frac{g^2}{2}-\frac{1}{2} \\& <\left\lceil\frac{g^2}{2}\right\rceil,
\end{align*}
 we have $\mathbf{P}(\phi(s),k)=*$ from \eqref{eqstar}.
\item If $h\in[\lceil\frac{g}{2}\rceil+1:g]$, then $s=<k-(g-h)(g+2)>_q$ from \eqref{constru2e}. Consequently, $$\phi(s)=<s-1>_q=<k-(g-h)(g+2)-1>_q.$$ Since
\begin{align*}
1&\leq(g-h)(g+2)+1 \\& \leq \left(g-\left\lceil\frac{g}{2}\right\rceil-1\right)(g+2)+1 \\& \leq\left(g-\frac{g}{2}-1\right)(g+2)+1 \\& =\frac{g^2}{2}-1 \\& <\left\lceil\frac{g^2}{2}\right\rceil,
\end{align*}
 we have $\mathbf{P}(\phi(s),k)=*$ from \eqref{eqstar}.
\end{itemize}
So the $\phi(s)$-th row is a star row for $s$.

For any $j\in[1:q]$, the set of integers $s$ satisfying $\phi(s)=j$ is denoted by $B_j$, then $B_j=\{<j+1>_q\}$. So $|B_j|=1$ for each $j\in[1:q]$. $\mathbf{P}$ satisfies the second item of Condition \ref{proper1}. Hence, $\mathbf{P}$ satisfies Condition \ref{proper1} with $\lambda=1$.
\end{proof}

 From Lemma \ref{lemma1}, for any positive integer $g$, there exists a $g$-$(\lceil\frac{g^2}{2}\rceil+g,\lceil\frac{g^2}{2}\rceil+g,\lceil\frac{g^2}{2}\rceil,\lceil\frac{g^2}{2}\rceil+g)$ PDA satisfying Condition \ref{proper1} with $\lambda=1$, then Theorem \ref{multibs} is proved from Corollary \ref{cothm1}.

\section{Proof of Proposition \ref{proP1k}}
\label{prproP1k}
\begin{proof}
Firstly, we will prove that the array $\mathbf{P}_1$ defined in \eqref{eqPDAtoDPDA} is a $(g-1)$-$(K_1,(g-1)F_1,(g-1)Z_1,gS_1)$ PDA.
\begin{itemize}
\item If $\mathbf{P}(j,k)=*$, we have $\mathbf{P}_1(j,k)=\mathbf{P}_1(j+F_1,k)=\ldots=\mathbf{P}_1(j+(g-2)F_1,k)=*$ from \eqref{eqPDAtoDPDA}. Since there are $Z_1$ stars in each column of $\mathbf{P}$, there are $Z=(g-1)Z_1$ stars in each column of $\mathbf{P}_1$. The condition C$1$ of Definition \ref{def-PDA} holds.
\item For any $s'\in[1:gS_1]$, let $<s'>_g=v$ and $\frac{s'-v}{g}+1=s$, then $s\in[1:S_1]$ and $s$ appears $g$ times in $\mathbf{P}$, assume that $\mathbf{P}(j_{s,1},k_{s,1})=\ldots=\mathbf{P}(j_{s,g},k_{s,g})=s$ with $j_{s,1}<\ldots<j_{s,g}$, then we have $T_s=\{(j_{s,1},k_{s,1}),\ldots,(j_{s,g},k_{s,g})\}$ from \eqref{Ts}. Consequently, we have
    \begin{alignat*}{3}
    &\mathbf{P}_1(j_{s,v},k_{s,v})=s',\\
    &\mathbf{P}_1(j_{s,v-1}+F_1,k_{s,v-1})=s',\\
    &\ \ \ \ \ \vdots \\
    &\mathbf{P}_1(j_{s,1}+(v-1)F_1,k_{s,1})=s',\\
    &\mathbf{P}_1(j_{s,g}+vF_1,k_{s,g})=s',\\
    &\mathbf{P}_1(j_{s,g-1}+(v+1)F_1,k_{s,g-1})=s',\\
    &\ \ \ \ \ \vdots \\
    &\mathbf{P}_1(j_{s,v+2}+(g-2)F_1,k_{s,v+2})=s'
    \end{alignat*}
    from \eqref{eqPDAtoDPDA}. That is, each integer in $[1:gS_1]$ occurs at least $g-1$ times in $\mathbf{P}_1$. The condition C2 of Definition \ref{def-PDA} holds. Moreover, since the average occurrence number of each integer in $\mathbf{P}_1$ is $\frac{K_1(g-1)(F_1-Z_1)}{gS_1}=g-1$, each integer in $[1:gS_1]$ occurs exactly $g-1$ times in $\mathbf{P}_1$.
    \item For any two distinct entries $\mathbf{P}_1(j_1,k_1)$ and $\mathbf{P}_1(j_2,k_2)$, if $\mathbf{P}_1(j_1,k_1)=\mathbf{P}_1(j_2,k_2)=s'\in[1:gS_1]$, we will show that $\mathbf{P}_1(j_1,k_2)=\mathbf{P}_1(j_2,k_1)=*$. Let $<s'>_g=v$ and $\frac{s'-v}{g}+1=s$, then we have $\mathbf{P}(<j_1>_{F_1},k_1)=\mathbf{P}(<j_2>_{F_1},k_2)=s$ from \eqref{eqPDAtoDPDA}, which implies $<j_1>_{F_1}=<j_2>_{F_1}, k_1=k_2$ or  $\mathbf{P}(<j_1>_{F_1},k_2)=\mathbf{P}(<j_2>_{F_1},k_1)=*$. If $<j_1>_{F_1}=<j_2>_{F_1}, k_1=k_2$ hold, let $T_s[\eta]=(<j_1>_{F_1},k_1)=(<j_2>_{F_1},k_2)$, $i_1=\lceil\frac{j_1}{F_1}\rceil$ and $i_2=\lceil\frac{j_2}{F_1}\rceil$, then we have $v=\Phi_{i_1-1}((1,2,\ldots,g))[\eta]=\Phi_{i_2-1}((1,2,\ldots,g))[\eta]$ from \eqref{eqPDAtoDPDA}, which implies $i_1=i_2$, i.e., $\lceil\frac{j_1}{F_1}\rceil=\lceil\frac{j_2}{F_1}\rceil$. Since $<j_1>_{F_1}=<j_2>_{F_1}$, we have $j_1=j_2$, which contradicts the hypothesis that $\mathbf{P}_1(j_1,k_1)$ and $\mathbf{P}_1(j_2,k_2)$ are two distinct entries. Hence, $\mathbf{P}(<j_1>_{F_1},k_2)=\mathbf{P}(<j_2>_{F_1},k_1)=*$ holds, which leads to $\mathbf{P}_1(j_1,k_2)=\mathbf{P}_1(j_2,k_1)=*$ from \eqref{eqPDAtoDPDA}. The condition C3 of Definition \ref{def-PDA} holds.
\end{itemize}
Hence, $\mathbf{P}_1$ is a $(g-1)$-$(K_1,(g-1)F_1,(g-1)Z_1,gS_1)$ PDA.

Next we will show that $\mathbf{P}_1$ satisfies Condition \ref{proper1} with $\lambda=g-1$. If $\mathbf{P}_1(j,k)=*$, we have $\mathbf{P}(<j>_{F_1},k)=*$ then  $\mathbf{P}_1(<j>_{F_1},k)=*$ from \eqref{eqPDAtoDPDA}. Conversely, if $\mathbf{P}_1(<j>_{F_1},k)=*$, we have $\mathbf{P}(<j>_{F_1},k)=*$ then $\mathbf{P}_1(j,k)=*$ from \eqref{eqPDAtoDPDA}. $\mathbf{P}_1$ satisfies the first item of Condition \ref{proper1}.

For any $s'\in[1:gS_1]$, let $v=<s'>_g$ and $s=\frac{s'-v}{g}+1$, then $v\in [1:g]$ and $s\in[1:S_1]$. Let $\phi(s')=j_{s,<v+1>_g}$, we will show that the $\phi(s')$-th row of $\mathbf{P}_1$ is a star row for $s'$.
For any $\mathbf{P}_1(j,k)=s'$ where $j\in[1:F_1]$ and $k\in[1:K_1]$, we have $\mathbf{P}(<j>_{F_1},k)=s$ from \eqref{eqPDAtoDPDA}. Let $(<j>_{F_1},k)=T_s[\eta]$, i.e., $k=k_{s,\eta}$, let $i=\lceil\frac{j}{F_1}\rceil$, we have $$v=\Phi_{i-1}((1,2,\ldots,g))[\eta]=\begin{cases} \eta+i-1, \ &\text{if} \ \eta\leq g-i+1,\\  \eta-g+i-1, \ &\text{otherwise},\end{cases}$$ which implies $<v+1>_g\neq \eta$ since $i\in[1:g-1]$. Consequently, we have $\mathbf{P}_1(\phi(s'),k)=\mathbf{P}_1(j_{s,<v+1>_g},k_{s,\eta})=*$ from the condition C3 of Definition \ref{def-PDA}. Hence, the $\phi(s')$-th row of $\mathbf{P}_1$ is a star row for $s'$.

For any $j\in[1:F_1]$, let $B_j=\{s'|\Phi(s')=j,s'\in[1:gS_1]\}$, we will show that $|B_j|=\frac{gS_1}{F_1}$.
For any $i\in[1:g-1]$, for any integer $s'$ in the $(j+(i-1)F_1)$-th row of $\mathbf{P}_1$, say $\mathbf{P}_1(j+(i-1)F_1,k)=s'$, let $v=<s'>_g$ and $s=\frac{s'-v}{g}+1$, then $\mathbf{P}(j,k)=s$ from \eqref{eqPDAtoDPDA}. Let $T_{s}[\eta]=(j,k)$,
$v'=<\eta-1>_g$, $\eta'$ be the integer satisfying $\Phi_{i-1}((1,2,\ldots,g))[\eta']=v'$ and $(j',k')=T_{s}[\eta']$, then $\mathbf{P}_1(j'+(i-1)F_1,k')=(s-1)g+v'$ from \eqref {eqPDAtoDPDA}. Moreover, we have $\phi((s-1)g+v')=j_{s, <v'+1>_g}=j_{s, \eta}=j$. It implies that for any integer $s'$ in the $(j+(i-1)F_1)$-th row of $\mathbf{P}_1$, there exists an integer $(s-1)g+v'$ such that $(s-1)g+v'\in B_j$.

For any $j\in[1:F_1]$ and any $i\in[1:g-1]$, for any two different integers in the $(j+(i-1)F_1)$-th row of $\mathbf{P}_1$, say $\mathbf{P}_1(j+(i-1)F_1,k_1)=s'_1$ and $\mathbf{P}_1(j+(i-1)F_1,k_2)=s'_2$, let $v_1=<s'_1>_g$, $s_1=\frac{s'_1-v_1}{g}+1$, $v_2=<s'_2>_g$ and $s_2=\frac{s'_2-v_2}{g}+1$, then we have
\begin{equation}
\label{s1s2}
\mathbf{P}(j,k_1)=s_1 \ \text{and} \ \mathbf{P}(j,k_2)=s_2
\end{equation}
from \eqref{eqPDAtoDPDA}. Let $T_{s_1}[\eta_1]=(j,k_1)$, $T_{s_2}[\eta_2]=(j,k_2)$, $v_1'=<\eta_1-1>_g$ and $v_2'=<\eta_2-1>_g$, then both $(s_1-1)g+v_1'$ and $(s_2-1)g+v_2'$ belong to $B_{j}$. Since $k_1\neq k_2$, we have $s_1\neq s_2$ from \eqref{s1s2} and the condition C3 of Definition \ref{def-PDA}. Without loss of generality, assume that $s_1<s_2$, then we have $(s_1-1)g+v_1'\neq (s_2-1)g+v_2'$. Otherwise if $(s_1-1)g+v_1'=(s_2-1)g+v_2'$, we have $v_1'=(s_2-s_1)g+v'_2\geq g+v'_2>g$, which contradicts $v'_1\in[1:g]$. It implies that different integers in the $(j+(i-1)F_1)$-th row correspond to different integers belonging to $B_j$.

On the other hand, since $\mathbf{P}$ satisfies Condition \ref{proper2}, i.e., each row has the same number of stars, assume that there are $t_1$ stars in each row of $\mathbf{P}$, then there are $K_1-t_1$ integers in each row of $\mathbf{P}$. Consequently, there are $K_1-t_1$ integers in each row of $\mathbf{P}_1$ from \eqref {eqPDAtoDPDA}. So we have $|B_j|\geq K_1-t_1$ for any $j\in[1:F_1]$. Since $\sum_{j=1}^{F_1}|B_j|=gS_1=\frac{(g-1)F_1(K_1-t_1)}{g-1}=F_1(K_1-t_1)$, we have $\frac{\sum_{j=1}^{F_1}|B_j|}{F_1}=K_1-t_1$, then $|B_j|=K_1-t_1=\frac{gS_1}{F_1}$. $\mathbf{P}_1$ satisfies the second item of Condition \ref{proper1}. Hence, $\mathbf{P}_1$ satisfies Condition \ref{proper1} with $\lambda=g-1$.

\end{proof}

\bibliographystyle{IEEEtran}
\bibliography{reference}

\end{document}